\documentclass[a4paper,11pt]{article}
\usepackage[utf8]{inputenc}

\usepackage{amsmath}
\usepackage{amssymb}
\usepackage{amsthm}
\usepackage{hyperref}
\usepackage{xspace}
\usepackage[in]{fullpage}
\usepackage{framed}
\usepackage{thm-restate}

\theoremstyle{plain}
\newtheorem{theorem}{Theorem}[section]
\newtheorem{lemma}[theorem]{Lemma}
\newtheorem{proposition}[theorem]{Proposition}
\newtheorem{corollary}[theorem]{Corollary}

\newtheorem{claim}[theorem]{Claim}
\theoremstyle{definition}
\newtheorem{definition}[theorem]{Definition}

\theoremstyle{remark}
\newtheorem{remark}[theorem]{Remark}

\usepackage[linesnumbered,noend,ruled,vlined]{algorithm2e}
\SetKwInput{KwInput}{Input}     
\SetKwInput{KwOutput}{Output}   

\newcommand{\yes}{\textsf{YES}}
\newcommand{\no}{\textsf{NO}}

\newcommand{\PSPACE}{\mathsf{PSPACE}}
\newcommand{\NP}{\mathsf{NP}}

\newcommand{\Z}{\mathbb{Z}}
\newcommand{\R}{\mathbb{R}}
\newcommand{\ang}[1]{\langle#1\rangle}
\newcommand{\ab}{\mathrm{ab}}

\newcommand{\DPR}{\textsc{Disjoint Paths Reconfiguration}\xspace}
\newcommand{\stDPR}{\textsc{Disjoint $s$-$t$ Paths Reconfiguration}\xspace}

\newcommand{\bd}{\partial}

\usepackage[mathlines]{lineno}

\usepackage{graphicx}

\usepackage{caption}
\title{%
Rerouting Planar Curves and Disjoint Paths\thanks{%
This work was supported by 
JSPS KAKENHI Grant Numbers JP18H05291, JP20K11692, JP20K14317, JP20H05793, JP20H05795, JP20K23323, JP21H03397, JP22K17854.}
}
\author{Takehiro Ito%
\thanks{Graduate School of Information Sciences, Tohoku University, Japan, {\tt takehiro@tohoku.ac.jp}}
\and 
Yuni Iwamasa%
\thanks{Graduate School of Informatics, Kyoto University, Japan, {\tt iwamasa@i.kyoto-u.ac.jp}}
\and 
Naonori Kakimura%
\thanks{Faculty of Science and Technology, Keio University, Japan, {\tt kakimura@math.keio.ac.jp}}
\and 
Yusuke Kobayashi%
\thanks{Research Institute for Mathematical Sciences, Kyoto University, Japan, {\tt yusuke@kurims.kyoto-u.ac.jp}}
\and 
Shun-ichi Maezawa%
\thanks{Department of Mathematics, Tokyo University of Science, Japan, {\tt maezawa.mw@gmail.com}}
\and
Yuta Nozaki%
\thanks{Graduate School of Advanced Science and Engineering, Hiroshima University, Japan, {\tt nozakiy@hiroshima-u.ac.jp}}
\and 
Yoshio Okamoto%
\thanks{Graduate School of Informatics and Engineering, The University of Electro-Communications, Japan, {\tt okamotoy@uec.ac.jp}}
\and
Kenta Ozeki%
\thanks{Faculty of Environment and Information Sciences, Yokohama National University, Japan, {\tt ozeki-kenta-xr@ynu.ac.jp}}
}

\date{}

\begin{document}

\maketitle

\begin{abstract}
In this paper, we consider a transformation of $k$ disjoint paths in a graph. For a graph and a pair of $k$ disjoint paths $\mathcal{P}$ and $\mathcal{Q}$ connecting the same set of terminal pairs, we aim to determine whether $\mathcal{P}$ can be transformed to $\mathcal{Q}$ by repeatedly replacing one path with another path so that the intermediates are also $k$ disjoint paths. The problem is called \textsc{Disjoint Paths Reconfiguration}. We first show that \textsc{Disjoint Paths Reconfiguration} is $\PSPACE$-complete even when $k=2$. On the other hand, we prove that, when the graph is embedded on a plane and all paths in $\mathcal{P}$ and $\mathcal{Q}$ connect the boundaries of two faces, \textsc{Disjoint Paths Reconfiguration} can be solved in polynomial time. The algorithm is based on a topological characterization for rerouting curves on a plane using the algebraic intersection number. We also consider a transformation of  disjoint $s$-$t$ paths as a variant. We show that the disjoint $s$-$t$ paths reconfiguration problem in planar graphs can be determined in polynomial time, while the problem is $\PSPACE$-complete in general. 
\end{abstract}

\section{Introduction}

\subsection{Disjoint Paths and Reconfiguration}

The \emph{disjoint paths problem} is a classical and important problem in algorithmic graph theory and combinatorial optimization. 
In the problem, the input consists of a graph $G=(V, E)$ and $2k$ distinct vertices $s_1, \dots , s_k, t_1, \dots , t_k$, called {\em terminals}, and the 
task is to find $k$ vertex-disjoint paths $P_1, \dots , P_k$ such that $P_i$ connects $s_i$ and $t_i$ for $i=1, \dots , k$ if they exist. 
A tuple $\mathcal{P}=(P_1, \dots , P_k)$ of paths satisfying this condition is called a \emph{linkage}. 
The disjoint paths problem has attracted attention since 1980s because of its practical applications to transportation networks, network routing~\cite{srinivas2005finding}, and VLSI-layout~\cite{frank,kramer1984complexity}. 
When the number $k$ of terminal pairs is part of the input, the disjoint paths problem was shown to be $\mathrm{NP}$-hard by Karp~\cite{karp1975computational}, 
and it remains $\mathrm{NP}$-hard even for planar graphs~\cite{lynch1975equivalence}. 
For the case of $k=2$, polynomial-time algorithms were presented in~\cite{Seymour,Shiloach,Thomassen},
while the directed variant was shown to be $\NP$-hard~\cite{fortune1980directed}. 
Later, for the case when the graph is undirected and $k$ is a fixed constant,
Robertson and Seymour~\cite{robertson1995graph} gave a polynomial-time algorithm based on the graph minor theory, 
which is one of the biggest achievements in this area. 
Although the setting of the disjoint paths problem is quite simple and easy to understand, 
a deep theory in discrete mathematics is required to solve the problem, 
which is a reason why this problem has attracted attention in the theoretical study of algorithms.

An interesting special case is the problem in planar graphs. 
In the early stages of the study of the disjoint paths problem, 
for the case when $G$ is embedded on a plane and all the terminals are on one face or two faces, 
polynomial-time algorithms were given in~\cite{robertson1986graph,suzuki1989algorithms,suzuki1990finding}. 
In the series of graph minor papers, 
the disjoint paths problem on a plane or on a fixed surface was solved for fixed $k$~\cite{robertson1988graph}. 
For the planar case, faster algorithms were presented in~\cite{ADLER2017815,Reed95,ReedRSS93}.  
The directed variant of the problem can be solved in polynomial time if the input digraph is planar and $k$ is a fixed constant~\cite{CyganMPP13,schrijverplanar}.

In this paper, we consider a transformation of linkages in a graph. 
Roughly, in a transformation, we pick up one path among the $k$ paths in a linkage, and replace it with another path to obtain a new linkage. 
To give a formal definition,  
suppose that $G$ is a graph and $s_1, \dots , s_k, t_1, \dots , t_k$ are distinct terminals. 
For two linkages $\mathcal{P} = (P_1, \dots , P_k)$ and $\mathcal{Q} = (Q_1, \dots , Q_k)$, 
we say that $\mathcal{P}$ is \emph{adjacent} to $\mathcal{Q}$ if 
there exists $i \in \{1, \dots , k\}$ such that 
$P_j = Q_j$ for $j \in \{1, \dots , k\} \setminus \{i\}$ and $P_i \not= Q_i$. 
We say that a sequence $\langle \mathcal{P}_1, \mathcal{P}_2, \dots, \mathcal{P}_\ell \rangle$ of linkages is a \emph{reconfiguration sequence} from $\mathcal{P}_1$ to $\mathcal{P}_\ell$
if 
$\mathcal{P}_i$ and $\mathcal{P}_{i+1}$ are adjacent 
for $i = 1, \dots , \ell -1$.
If such a sequence exists, we say that $\mathcal{P}_1$ is \emph{reconfigurable} to $\mathcal{P}_\ell$. 
In this paper, we focus on the following reconfiguration problem, which we call \DPR. 

\vspace{-5pt}

\begin{framed}
\noindent
\DPR

\noindent
\textbf{Input.} A graph $G=(V, E)$, distinct terminals $s_1, \dots , s_k, t_1, \dots , t_k$,  
and two linkages $\cal P$ and $\cal Q$. 

\noindent
\textbf{Question.}
Is $\cal P$ reconfigurable to $\cal Q$?
\end{framed}

\vspace{-5pt}

The problem can be regarded as the problem of deciding the reachability between linkages via rerouting paths.
Such a problem falls in the area of \textit{combinatorial reconfiguration};
see Section~\ref{sec:related} for prior work on combinatorial reconfiguration.
Note that \DPR is a decision problem that just returns ``\yes'' or ``\no'' and does not necessarily find a reconfiguration sequence when the answer is \yes.

Although our study is motivated by theoretical interest in the literature of combinatorial reconfiguration,
the problem can model rerouting problem in a telecommunication network as follows.
Suppose that a linkage represents routing in a telecommunication network, and we want to modify linkage $\mathcal{P}$ to another linkage $\mathcal{Q}$ which is better than $\mathcal{P}$ in some sense.  
If we can change only one path in a step in the network for some technical reasons, and we have to keep a linkage in the modification process, 
then this situation is modeled as \DPR.

We also study a special case of the disjoint paths problem when $s_1 = \dots = s_k$ and $t_1 = \dots = t_k$, which we call the \emph{disjoint $s$-$t$ paths problem}. 
In the problem, for a graph and two terminals $s$ and $t$, we seek for $k$ internally vertex-disjoint (or edge-disjoint) paths connecting $s$ and $t$. 
It is well-known that the disjoint $s$-$t$ paths problem can be solved in polynomial time. 
The study of disjoint $s$-$t$ paths was originated from Menger's min-max theorem~\cite{Menger27} and the max-flow algorithm by Ford and Fulkerson~\cite{FordF56}. 
Faster algorithms for finding maximum disjoint $s$-$t$ paths or a maximum $s$-$t$ flow have been actively studied in particular for planar graphs; see e.g.~\cite{DBLP:conf/isaac/EnochFMM21,DBLP:journals/algorithmica/KaplanN11,DBLP:journals/algorithmica/KhullerN94,10.1145/3504032}.

In the same way as \DPR, we consider a reconfiguration of internally vertex-disjoint $s$-$t$ paths. 
Let $G=(V, E)$ be a graph with two distinct terminals $s$ and $t$. 
We say that a set $\mathcal{P}=\{P_1, \dots , P_k\}$ of $k$ paths in $G$ is an \emph{$s$-$t$ linkage} if 
$P_1, \dots , P_k$ are internally vertex-disjoint $s$-$t$ paths. 
Note that $\mathcal{P}$ is not a tuple but a set, that is, we ignore the ordering of the paths in $\mathcal{P}$.  
We say that $s$-$t$ linkages $\mathcal{P}$ and $\mathcal{Q}$ are \emph{adjacent} if $\mathcal{Q} = (\mathcal{P} \setminus P) \cup \{Q\}$ for some $s$-$t$ paths $P$ and $Q$ with $P\neq Q$. 
We define the reconfigurability of $s$-$t$ linkages in the same way as linkages. 
We consider the following problem.

\vspace{-5pt}

\begin{framed}
\noindent
\stDPR 

\noindent
\textbf{Input.} A graph $G=(V, E)$, distinct terminals $s$ and $t$, and two $s$-$t$ linkages $\cal P$ and $\cal Q$.

\noindent
\textbf{Question.}
Is $\cal P$ reconfigurable to $\cal Q$?
\end{framed}

\vspace{-5pt}

\subsection{Our Contributions}
\label{sec:contribution}

Since finding disjoint $s$-$t$ paths is an easy combinatorial optimization problem,  
one may expect that \stDPR is also tractable. However, this is not indeed the case. 
We show that \stDPR is $\PSPACE$-hard even when $k=2$. 

\begin{restatable}{theorem}{hardness}
\label{thm:pspacecompl-twopaths}
The \stDPR is $\PSPACE$-complete even when $k=2$ and the maximum degree of $G$ is four.
\end{restatable}

Note that \stDPR can be easily reduced to \DPR by splitting each of $s$ and $t$ into $k$ terminals. 
Thus, this theorem implies the $\PSPACE$-hardness of \DPR with $k =2$.

In this paper, we mainly focus on the problems in planar graphs. 
To better understand \DPR in planar graphs, we show a topological necessary condition. 

Topological conditions play important roles in the disjoint paths problem. 
If there exist disjoint paths connecting terminal pairs in a graph embedded on a surface $\Sigma$, then 
obviously there must exist disjoint curves on $\Sigma$ connecting them.  
For example, when terminals $s_1, s_2, t_1$ and $t_2$ lie on the outer face $F$ in a plane graph $G$ in this order, 
there exist no disjoint curves connecting the terminal pairs in the disk $\Sigma = \R^2 \setminus F$, 
and hence we can conclude that $G$ contains no disjoint paths. 
Such a topological condition is used to design polynomial-time algorithms for the disjoint paths problem with $k=2$~\cite{Seymour,Shiloach,Thomassen}, 
and to deal with the problem on a disk or a cylinder~\cite{robertson1986graph}. 
When $\Sigma$ is a plane (or a sphere), we can always connect terminal pairs by disjoint curves on $\Sigma$, and hence 
nothing is derived from the above argument. 
Indeed, Robertson and Seymour~\cite{robertson1988graph} showed that 
if the input graph is embedded on a surface and the terminals are mutually ``far apart,'' then
desired disjoint paths always exist. 

In contrast, as we will show below in Theorem~\ref{thm:genus0}, there exists a topological necessary condition for the reconfigurability of disjoint paths. 
Thus, even when the terminals are mutually far apart, the reconfiguration of disjoint paths is not always possible. 
This shows a difference between the disjoint paths problem and \DPR.

In order to formally discuss the topological necessary condition, we consider the reconfiguration of curves on a surface.  
Suppose that $\Sigma$ is a surface and let $s_1, \dots , s_k, t_1, \dots , t_k$ be distinct points on $\Sigma$. 
By abuse of notation, we say that $\mathcal{P} = (P_1, \dots , P_k)$ is a \emph{linkage} if it is a collection of disjoint simple curves on $\Sigma$ such that $P_i$ connects $s_i$ and $t_i$. 
We also define the adjacency and reconfiguration sequences for linkages on $\Sigma$ in the same way as linkages in a graph.  
Then, the reconfigurability between two linkages on a plane can be characterized  
with a word $w_j$ associated to $Q_j$ which is an element of the free group $F_k$ generated by $x_1,\dots,x_k$
as follows; see Section~\ref{sec:surface} for the definition of $w_j$.

\begin{restatable}{theorem}{reconfcurve}
\label{thm:genus0}
Let $\mathcal{P} = (P_1, \dots , P_k)$ and $\mathcal{Q} = (Q_1, \dots , Q_k)$ be linkages on a plane \textup{(}or a sphere\textup{)}.
Then, $\mathcal{P}$ is reconfigurable to $\mathcal{Q}$ if and only if $w_j \in \ang{x_j}$ for any $j \in \{1, \dots , k\}$, where $\ang{x_j}$ denotes the subgroup generated by $x_j$.
\end{restatable}

See Figure~\ref{fig:noinstance01} (left) for an example. 
It is worth noting that, if $k=2$ and $\Sigma$ is a connected orientable closed surface of genus $g \ge 1$, then 
such a topological necessary condition does not exist, i.e., the reconfiguration is always possible; 
see Appendix~\ref{sec:surfaceappendix}.

\begin{figure}[t]
    \centering
    \includegraphics[width=10cm]{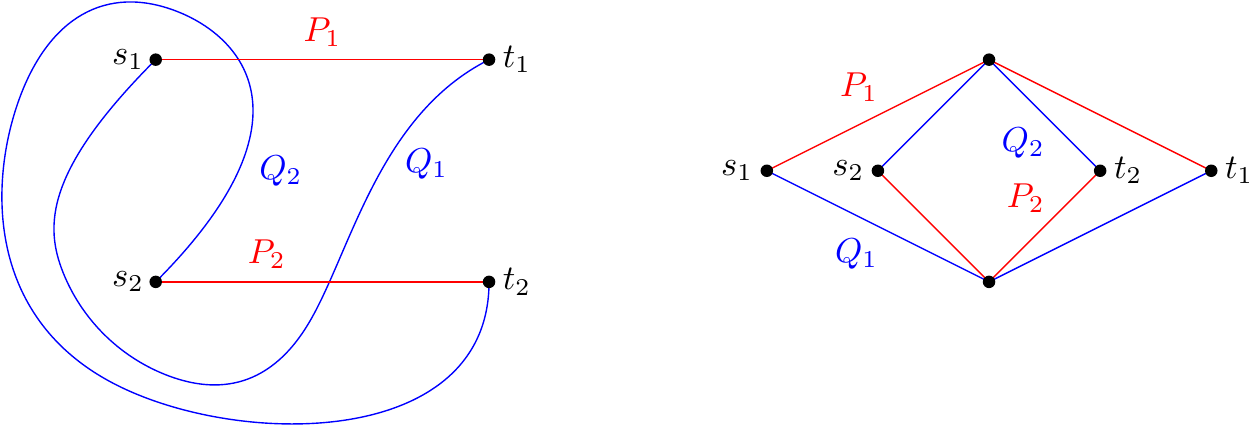}
    \caption{(Left) An example on the plane where $(P_1, P_2)$ is not reconfigurable to $(Q_1, Q_2)$.
    (Right) An example in a graph where the condition in Theorem~\ref{thm:genus0} holds but $(P_1, P_2)$ is not reconfigurable to $(Q_1, Q_2)$.}
    \label{fig:noinstance01}
\end{figure}

For a graph embedded on a plane, we can identify paths and curves. 
Then, Theorem~\ref{thm:genus0} gives a topological necessary condition for \DPR in planar graphs. 
However, the converse does not necessarily hold: even when the condition in Theorem~\ref{thm:genus0} holds, an instance of \DPR may have no reconfiguration sequence.
See \figurename~\ref{fig:noinstance01} (right) for a simple example.
The polynomial solvability of \DPR in planar graphs is open even for the case of $k=2$.

With the aid of the topological necessary condition,
we design polynomial-time algorithms for special cases, in which all the terminals are on a single face (called \emph{one-face instances}), or $s_1, \dots , s_k$ are on some face and $t_1, \dots , t_k$ are on another face (called \emph{two-face instances}). 
Note that one/two-face instances have attracted attention in the disjoint paths problem~\cite{robertson1986graph,suzuki1989algorithms,suzuki1990finding}, 
in the multicommodity flow problem~\cite{OKAMURA198355,okamura1981multicommodity}, and
in the shortest disjoint paths problem~\cite{DBLP:conf/esa/BorradaileNZ15,verdiere2011shortest,datta2018shortest,kobayashi2010shortest}.
We show that any one-face instance of \DPR has a reconfiguration sequence (Proposition~\ref{prop:1face}). 
Moreover, we prove a topological characterization for two-face instances of \DPR~(Theorem~\ref{thm:2facechara}), which leads to a polynomial-time algorithm in this case.

\begin{restatable}{theorem}{twoface}
\label{thm:2facecalgo}
When the instances are restricted to two-face instances,
\DPR can be solved in polynomial time.  
\end{restatable}

Based on this theorem, we give a polynomial-time algorithm for \stDPR in planar graphs. 

\begin{restatable}{theorem}{stalgorithm}
\label{thm:stplanaralgo}
There is a polynomial-time algorithm for \stDPR in planar graphs. 
\end{restatable}

Note that the number $k$ of paths in Theorems~\ref{thm:2facecalgo} and~\ref{thm:stplanaralgo} can be part of the input.

It is well known that $G$ has an $s$-$t$ linkage of size $k$ if and only if $G$ has no $s$-$t$ separator of size $k-1$~(Menger's theorem).
The characterization for two-face instances (Theorem~\ref{thm:2facechara}) implies 
the following theorem, which is interesting in the sense that one extra $s$-$t$ connectivity is sufficient to guarantee the existence of a reconfiguration sequence.

\begin{restatable}{theorem}{stcharacterization}
\label{thm:stplanarchara}
Let $G=(V, E)$ be a planar graph with distinct vertices $s$ and $t$, and 
let $\mathcal{P}$ and $\mathcal{Q}$ be $s$-$t$ linkages. 
If there is no $s$-$t$ separator of size $k$, then $\mathcal{P}$ is reconfigurable to $\mathcal{Q}$. 
\end{restatable}

As mentioned above, the polynomial solvability of \DPR in planar graphs is open even for the case of $k=2$. 
On the other hand, when $k$ is not bounded, \DPR is $\PSPACE$-complete as the next theorem shows.

\begin{restatable}{theorem}{planarhardness}
\label{thm:pspacecompl-planar}
The \DPR is $\PSPACE$-complete when the graph $G$ is planar and of bounded bandwidth.
\end{restatable}

Here, we recall the definition of the bandwidth of a graph.
Let $G=(V, E)$ be an undirected graph.
Consider an injective map $\pi\colon V \to \Z$.
Then, the \emph{bandwidth} of $\pi$ is defined as $\max \{|\pi(u)-\pi(v)| \mid \{u, v\} \in E\}$.
The \emph{bandwidth} of $G$ is defined as the minimum bandwidth of all injective maps $\pi\colon V \to \Z$.

\subsection{Related Work}
\label{sec:related}

\emph{Combinatorial reconfiguration} is an emerging field in discrete mathematics and theoretical computer science.
In typical problems of combinatorial reconfiguration, we consider two discrete structures, and ask whether one can be transformed to the other by a sequence of local changes.
See surveys of Nishimura~\cite{DBLP:journals/algorithms/Nishimura18} and van den Heuvel \cite{DBLP:books/cu/p/Heuvel13}.

Path reconfiguration problems have been studied in this framework.
The apparently first problem is the shortest path reconfiguration, introduced by
Kaminski et al.~\cite{DBLP:journals/tcs/KaminskiMM11}.
In this problem, we are given an undirected graph with two designated vertices $s$, $t$ and two $s$-$t$ shortest paths $P$ and $Q$.
Then, we want to decide whether $P$ can be transformed to $Q$ by a sequence of one-vertex changes in such a way that all the intermediate $s$-$t$ paths remain the shortest.
Bonsma~\cite{DBLP:journals/tcs/Bonsma13} proved that the shortest path reconfiguration is $\PSPACE$-complete, but polynomial-time solvable when the input graph is chordal or claw-free.
Bonsma~\cite{DBLP:journals/dam/Bonsma17} further proved that the problem is polynomial-time solvable for planar graphs.
Wrochna \cite{DBLP:journals/jcss/Wrochna18} proved that the problem is $\PSPACE$-complete even for graphs of bounded bandwidth.
Gajjar et al.~\cite{DBLP:journals/corr/abs-2112-07499} proved that the problem is polynomial-time solvable for circle graphs, circular-arc graphs, permutation graphs and hypercubes.
They also considered a variant where a change can involve $k$ successive vertices; in this variant they proved that the problem is $\PSPACE$-complete even for line graphs.
Properties of the adjacency relation in the shortest path reconfiguration have also been studied~\cite{DBLP:journals/dm/AsplundEHHNW18, DBLP:journals/dm/AsplundW20}.

Another path reconfiguration problem has been introduced by Amiri et al.~\cite{DBLP:conf/icalp/AmiriDSW18} who were motivated by a problem in software defined networks.
In their setup, we are given a directed graph with edge capacity and two designated vertices $s,t$.
We are also given $k$ pairs of $s$-$t$ paths $(P_i, Q_i)$, $i=1,2,\dots,k$, where the number of paths among $P_1, P_2, \dots, P_k$ (and among $Q_1, Q_2, \dots, Q_k$ respectively) traversing an edge is at most the capacity of the edge.
The problem is to determine whether one set of paths can be transformed to the other set of paths by a sequence of the following type of changes: specify one vertex $v$ and then switch the usable outgoing edges at $v$ from those in the $P_i$ to those in the $Q_i$.
In each of the intermediate situations, there must be a unique path through usable edges in $P_i \cup Q_i$ for each $i$.
See \cite{DBLP:conf/icalp/AmiriDSW18} for the precise problem specification.
Amiri et al.~\cite{DBLP:conf/icalp/AmiriDSW18} proved that the problem is $\NP$-hard even when $k=2$.
For directed acyclic graphs, they also proved that the problem is $\NP$-hard~(for unbounded $k$) but fixed-parameter tractable with respect to $k$.
A subsequent work~\cite{DBLP:conf/networking/AmiriDP0W19} studied an optimization variant in which the number of steps is to be minimized when a set of ``disjoint'' changes can be performed simultaneously.

Matching reconfiguration in bipartite graphs can be seen as a certain type of disjoint paths reconfiguration problems.
In matching reconfiguration, we are given two matchings~(with extra properties) and want to determine whether one matching can be transformed to the other matching by a sequence of local changes.
There are several choices for local changes.
One of the most studied local change rules is the token jumping rule, where we remove one edge and add one edge at the same time.
Ito et al.~\cite{DBLP:journals/tcs/ItoDHPSUU11} proved that the matching reconfiguration~ (under the token jumping rule) can be solved in polynomial time.\footnote{The theorem by Ito et al.~\cite{DBLP:journals/tcs/ItoDHPSUU11} only gave a polynomial-time algorithm for a different local change, 
the so-called token addition and removal rule. However, their result can easily be adapted to the token jumping rule, too. See~\cite{DBLP:journals/jco/ItoKKKO19}.} 

To see a connection of matching reconfiguration with disjoint paths reconfiguration, consider the matching reconfiguration problem in bipartite graphs $G$ under the token jumping rule, where we are given two matchings $M, M'$ of $G$.
Then, we add two extra vertices $s, t$ to $G$, and for each edge $e \in M$ (and $M'$) we construct a unique $s$-$t$ path of length three that passes through $e$.
This way, we obtain two $s$-$t$ linkages $\mathcal{P}$ and $\mathcal{P}'$ from $M$ and $M'$, respectively.
It is easy to observe that $\mathcal{P}$ can be reconfigured to $\mathcal{P}'$ in \stDPR if and only if $M$ can be reconfigured to $M'$ in the matching reconfiguration problem in $G$.

There are a lot of studies on the disjoint paths problem and its variants. 
A natural variant is to maximize the number of vertex-disjoint paths connecting terminal pairs, which is called the \emph{maximum disjoint paths problem}. 
Since this problem is $\NP$-hard when $k$ is part of the input, it has been studied from the viewpoint of approximation algorithms. 
It is known that a simple greedy algorithm achieves an approximation of factor $O(\sqrt{|V|})$~\cite{DBLP:journals/mp/KolliopoulosS04}, 
which is the current best approximation ratio for general graphs. 
Chuzhoy and Kim~\cite{DBLP:conf/approx/ChuzhoyK15} improved this factor to $\tilde O(n^{1/4})$ for grid graphs, and 
Chuzhoy et al.~\cite{DBLP:conf/stoc/ChuzhoyKL16} gave an $\tilde O(|V|^{9/19})$-approximation algorithm for planar graphs.  
On the negative side, 
the maximum disjoint paths problem is $2^{\Omega(\sqrt{\log |V|})}$-hard to approximate under some complexity assumption~\cite{DBLP:journals/siamcomp/ChuzhoyKN22}, 
which is the current best hardness result.

Another variant is the \emph{shortest non-crossing walks problem}. 
In a graph embedded on a plane, we say that walks are \emph{non-crossing} if 
they do not cross each other or themselves, while they may share edges or vertices; see~\cite{DBLP:conf/soda/EricksonN11} for a formal definition. 
We can see that this concept is positioned in between disjoint paths and disjoint curves. 
In the shortest non-crossing walks problem, we are given a plane graph with non-negative edge lengths and $k$ terminal pairs that 
lie on the boundary of $h$ polygonal obstacles. 
The objective is to find $k$ non-crossing walks that connect terminal pairs in $G$ of minimum total length, if they exist. 
For the case of $h=2$, Takahashi et al.~\cite{DBLP:conf/isaac/TakahashiSN92} gave an $O(|V| \log |V|)$-time algorithm, 
and Papadopoulou~\cite{DBLP:journals/ijcga/Papadopoulou99} proposed a linear-time algorithm.  
For general $h$, Erickson and Nayyeri~\cite{DBLP:conf/soda/EricksonN11} gave a $2^{O(h^2)} |V| \log k$-time algorithm, 
which runs in polynomial time when $h$ is a fixed constant. 
They also showed that the existence of a feasible solution can be determined in linear time.


\subsection{Organization}

In Section~\ref{sec:preliminaries}, we introduce some notation and 
basic concepts in topology. 
Section~\ref{sec:surface} deals with rerouting disjoint curves, giving the proof of Theorem~\ref{thm:genus0}.
In Sections~\ref{sec:planar_graph} and~\ref{sec:proof2face}, we prove Theorems~\ref{thm:2facecalgo}, \ref{thm:stplanaralgo}, and \ref{thm:stplanarchara}.
Hardness results~(Theorems~\ref{thm:pspacecompl-twopaths} and \ref{thm:pspacecompl-planar}) are then proven in Section~\ref{sec:pspacecompl}.

\section{Preliminaries}
\label{sec:preliminaries}

For a positive integer $k$, let $[k] = \{1, 2, \dots , k\}$. 

Let $G=(V, E)$ be a graph. For a subgraph $H$ of $G$, the vertex set of $H$ is denoted by $V(H)$. 
Similarly, for a path $P$, let $V(P)$ denote the set of vertices in $P$.  
For $X \subseteq V$, let $N(X)$ be the set of vertices in $V \setminus X$ that are adjacent to the vertices in $X$. 
For a vertex set $U \subseteq V$, let $G \setminus U$ denote the graph obtained from $G$ by removing all the vertices in $U$ and the incident edges. 
For a path $P$ in $G$, we denote $G \setminus V(P)$ by $G \setminus P$ to simplify the notation. 
For disjoint vertex sets $X, Y \subseteq V$, 
we say that a vertex subset $U \subseteq V \setminus (X \cup Y)$ \emph{separates} $X$ and $Y$ if 
$G \setminus U$ contains no path between $X$ and $Y$. 
For distinct vertices $s, t \in V$, $U \subseteq V \setminus \{s, t\}$ is called an \emph{$s$-$t$ separator} if $U$ separates $\{s\}$ and $\{t\}$.

For \DPR (resp.~\stDPR), an instance is denoted by a triplet $(G, \mathcal{P}, \mathcal{Q})$, 
where $G$ is a graph and $\mathcal{P}$ and $\mathcal{Q}$ are linkages (resp.~$s$-$t$ linkages). 
Note that we omit the terminals, because they are determined by $\mathcal{P}$ and $\mathcal{Q}$. 
Since any instance has a trivial reconfiguration sequence when $k=1$, we may assume that $k\ge 2$.  
For linkages (resp.~$s$-$t$ linkages) $\mathcal{P}$ and $\mathcal{Q}$, 
we denote $\mathcal{P} \leftrightarrow \mathcal{Q}$ if $\mathcal{P}$ and $\mathcal{Q}$ are adjacent. 
Recall that $\mathcal{P} = (P_1, \dots , P_k)$ is \emph{adjacent} to $\mathcal{Q}= (Q_1, \dots , Q_k)$ if 
there exists $i \in [k]$ such that 
$P_j = Q_j$ for $j \in [k] \setminus \{i\}$ and $P_i \not= Q_i$.

For a graph $G$ embedded on a surface $\Sigma$, each connected region of $\Sigma \setminus G$ is called a \emph{face} of $G$. 
For a face $F$, its boundary is denoted by $\bd{F}$. 
When a graph $G$ is embedded on a surface $\Sigma$, 
a path in $G$ is sometimes identified with the corresponding curve in $\Sigma$. 
A graph embedded on a plane is called a \emph{plane graph}. 
A graph is said to be \emph{planar} if it has a planar embedding.

The following notion is well known in topology.
See \cite[Section~1.2.3]{FaMa12} for instance.

\begin{definition}
\label{def:local_int}
Let $C_1$ and $C_2$ be piecewise smooth oriented curves on an oriented surface and let $p \in C_1\cap C_2$ be a transverse double point.
The \emph{local intersection number} $\varepsilon_p(C_1,C_2)$ of $C_1$ and $C_2$ at $p$ is defined by $\varepsilon_p(C_1,C_2)=1$ if $C_1$ crosses $C_2$ from left to right and $\varepsilon_p(C_1,C_2)=-1$ if $C_1$ crosses $C_2$ from right to left (see Figure~\ref{fig:loc_int_num}).
When $\partial C_1\cap C_2 = C_1\cap\partial C_2 = \emptyset$, the \emph{algebraic intersection number} $\mu(C_1,C_2) \in \Z$ is defined to be the sum of $\varepsilon_p(C_1,C_2)$ over all $p \in C_1\cap C_2$ (after a small perturbation if necessary).
Note that $\partial C_i$ denotes the set of endpoints of $C_i$. 
\end{definition}

\begin{figure}[t]
    \centering
    \includegraphics{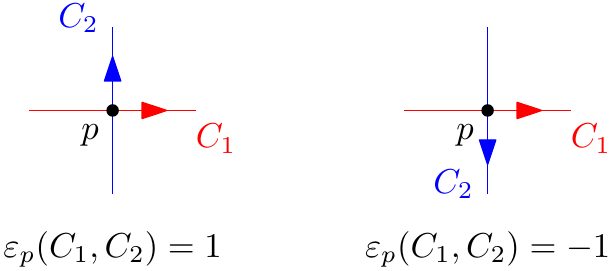}
    \caption{Local intersection numbers of curves $C_1$ and $C_2$ at $p$.}
    \label{fig:loc_int_num}
\end{figure}

\begin{figure}[t]
    \centering
    \includegraphics[width=10cm]{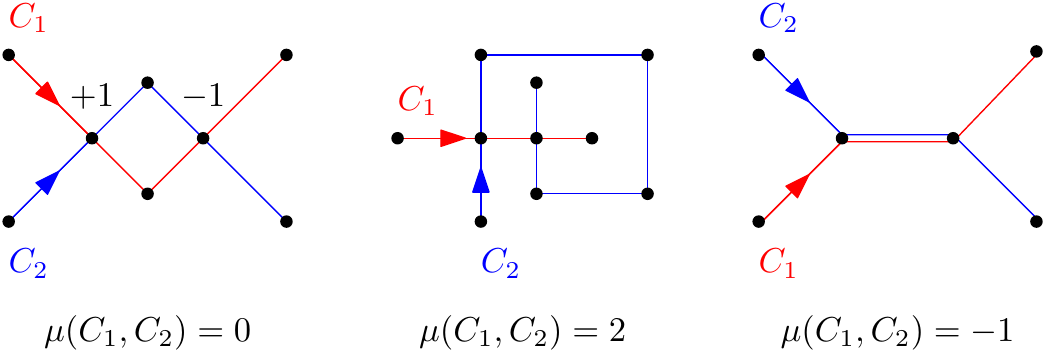}
    \caption{Algebraic intersection numbers of paths $C_1$ and $C_2$ on a graph.}
    \label{fig:alg_int_num}
\end{figure}


When a graph is embedded on an oriented surface, paths in the graph are piecewise smooth curves, and hence 
we can define the algebraic intersection number for a pair of paths (see Figure~\ref{fig:alg_int_num}).

\section{Curves on a Plane}
\label{sec:surface}

In this section, we consider the reconfiguration of curves on a plane and prove Theorem~\ref{thm:genus0}.
Suppose that we are given distinct points $s_1, \dots , s_k, t_1, \dots , t_k$ on a plane
and linkages $\mathcal{P}$ and $\mathcal{Q}$ that consist of curves on the plane connecting $s_i$ and $t_i$. 

Throughout this section, all intersections of curves are assumed to be transverse double points.
Fix $j\in[k]$ and let $\bigcup_{i\in[k]}P_i\cap Q_j = \{s_j,p_1,\dots,p_n,t_j\}$, where the $n+2$ points are aligned on $Q_j$ in this order.
We now define $w_j \in F_k$ by
\[
w_j=\prod_{\ell\in[n]} x_{i_\ell}^{\varepsilon_{p_\ell}(P_{i_\ell}, Q_j)},
\]
where $i_\ell\in[k]$ satisfies $p_\ell \in P_{i_\ell}\cap Q_j$.
Recall that $F_k$ denotes the free group generated by $x_1,\dots,x_k$.
We give an example in Figure~\ref{fig:counterexample}.

\begin{figure}[t]
    \centering
    \includegraphics{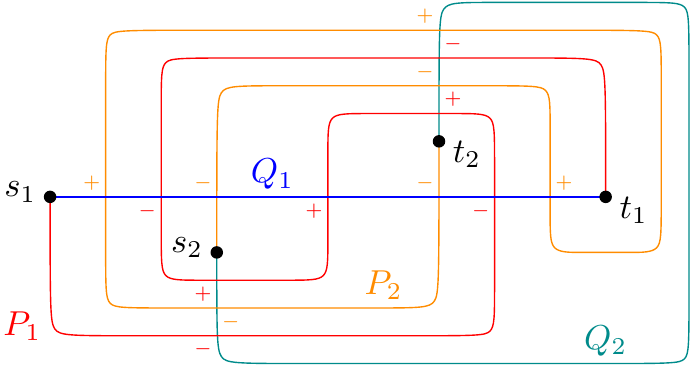}
    \caption{An example of linkages with $w_1=x_2x_1^{-1}x_2^{-1}x_1x_2^{-1}x_1^{-1}x_2$ and $w_2=x_1x_2^{-1}x_1^{-1}x_2x_1^{-1}x_2^{-1}x_1$.}
    \label{fig:counterexample}
\end{figure}

\begin{remark}
\label{rem:algintnum}
Let $\ab\colon F_k \to \Z^k$ denote the abelianization, that is, the $\ell$th entry of $\ab(w)$ is the sum of the exponents of $x_\ell$'s in $w$.
For distinct $i,j \in [k]$, the $i$th entry of $\ab(w_j)$ is equal to the algebraic intersection number $\mu(P_i,Q_j)\in \Z$ of $P_i$ and $Q_j$. 
Thus, $w_j \in \ang{x_j}$ implies that $\mu(P_i,Q_j)=0$ for any $i \in [k] \setminus \{j\}$.
\end{remark}

In the following two lemmas, we observe the behavior of $w_j$ under certain moves of curves.
For $j \in [k]$, let $w'_j$ denote the word defined by a linkage $\mathcal{P}'$ and the curve $Q_j$. 

\begin{lemma}
\label{lem:isotopy}
Let $i\in[k]$ and let $\mathcal{P}' = (P'_1, \dots , P'_k)$ be a linkage such that $P'_\ell=P_\ell$ if $\ell\neq i$, and $P'_i$ is isotopic to $P_i$ relative to $\{s_i,t_i\}$ in $\R^2\setminus \bigcup_{\ell\neq i}P_\ell$.
Then, $w'_j = w_j$ for $j \in [k] \setminus \{i\}$, and $w'_i = x_i^{e_1}w_i x_i^{e_2}$ for some $e_1,e_2\in\Z$.
\end{lemma}

\begin{proof}
By the definition of an isotopy (see \cite[Section~1.2.5]{FaMa12}), $P'_i$ is obtained from $P_i$ by a finite sequence of the moves illustrated in Figure~\ref{fig:isotopy}.
By (I), one intersection of $P_i$ and $Q_i$ is created or eliminated, and thus (I) changes $w_i$ to $w_ix_i^{\pm 1}$ or $x_i^{\pm 1}w_i$.
In (II), two intersections of $P_i$ and $Q_\ell$ are created or eliminated for some $\ell \in [k]$.
Since $x_i^{\pm 1}x_i^{\mp 1}=1$, $w_j$ is unchanged under (II) for any $j \in [k]$.
\end{proof}

\begin{figure}[t]
    \centering
    \includegraphics{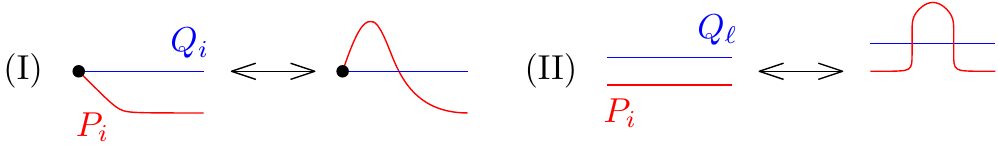}
    \caption{Local pictures of isotopies of $P_i$.}
    \label{fig:isotopy}
\end{figure}

Recall here that $\ang{x_\ell}$ denotes the subgroup of $F_k$ generated by $x_\ell$.

\begin{lemma}
\label{lem:slide}
Let $\gamma$ be a simple curve connecting $P_i$ and $s_j$ $(i\neq j)$ whose interior is disjoint from $\bigcup_{\ell \in [k]} P_\ell$, and define $P'_i$ as illustrated in Figure~\ref{fig:slide_gamma}.
Let $\mathcal{P}'$ be the linkage obtained from $\mathcal{P}$ by replacing $P_i$ with $P'_i$. 
For $\ell \in [k]$, if $w_\ell \in \ang{x_\ell}$, then $w'_\ell=w_\ell$.
\end{lemma}

\begin{proof}
Define a group homomorphism $f_{ij}\colon F_k \to F_k$ by $f_{ij}(x_\ell)=x_\ell$ if $\ell\neq j$, and $f_{ij}(x_j)=x_ix_jx_i^{-1}$.
Then, one can check that $w'_\ell=f_{ij}(w_\ell)$ if $\ell\neq j$, and $w'_j=x_i^{-1}f_{ij}(w_j)x_i$ (see Figure~\ref{fig:slide_gamma}).
Since $w_\ell=x_\ell^{e_\ell}$ for some $e_\ell \in \Z$ by the assumption, we have $w'_\ell=w_\ell$ if $\ell\neq j$.
Also, one has
\[
w'_j=x_i^{-1}f_{ij}(w_j)x_i=x_i^{-1}(x_ix_jx_i^{-1})^{e_j}x_i=w_j.
\]
This completes the proof.
\end{proof}

\begin{figure}[t]
    \centering
    \includegraphics{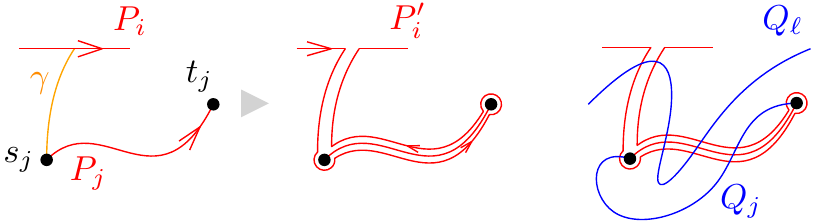}
    \caption{(Left) A move of $P_i$ along $\gamma$. (Right) Intersections of $P'_i$ and $\bigcup_\ell Q_\ell$.}
    \label{fig:slide_gamma}
\end{figure}

As a consequence of Lemmas~\ref{lem:isotopy} and \ref{lem:slide}, we obtain the following key lemma.

\begin{lemma}
\label{lem:invariant}
Suppose that $\mathcal{P}$ is reconfigurable to $\mathcal{P}'$. 
For $j \in [k]$, if $w_j \in \ang{x_j}$, then $w'_j \in \ang{x_j}$. 
\end{lemma}

\begin{proof}
It suffices to consider the case when there is $i\in[k]$ such that $P'_\ell = P_\ell$ if $\ell\neq i$, and $P'_i \neq P_i$.
Since $P_i$ is isotopic to $P'_i$ (relative to $\{s_i,t_i\}$) in $\R^2$, the curve $P'_i$ is obtained from $P_i$ by the moves in Lemmas~\ref{lem:isotopy} and \ref{lem:slide}.
Therefore, these lemmas imply that if $w_j \in \ang{x_j}$ then $w'_j \in \ang{x_j}$.
\end{proof}

With this key lemma, we can prove Theorem~\ref{thm:genus0} stating that 
$\mathcal{P}$ is reconfigurable to $\mathcal{Q}$ if and only if $w_j \in \ang{x_j}$ for any $j \in [k]$.

\begin{proof}[Proof of Theorem~\ref{thm:genus0}]
First suppose that $\mathcal{P}$ is reconfigurable to $\mathcal{Q}$, namely $\mathcal{P}$ is reconfigurable to $\mathcal{P}'$ such that $P'_i\cap Q_i=\{s_i,t_i\}$ and $P'_i\cap Q_j=\emptyset$ for $j \in [k] \setminus \{i\}$.
Then, $w'_j=1$ for any $j \in [k]$.
Since $\mathcal{P}'$ is reconfigurable to $\mathcal{P}$, Lemma~\ref{lem:invariant} implies that $w_j \in \ang{x_j}$ for any $j \in [k]$.

The converse is shown by induction on the number, say $n$, of intersections of $\mathcal{P}$ and $\mathcal{Q}$ except their endpoints.
The case $n=0$ is obvious.
Let us consider the case $n\geq 1$. 
If $P_i \cap Q_j = \emptyset$ for any pair of distinct $i, j \in [k]$, then the reconfiguration is obviously possible. 
Otherwise, there exists $x_ix_i^{-1}$ or $x_i^{-1}x_i$ in the product of the definition of $w_{j^*}$ 
for some $i,j^* \in [k]$ (possibly $i=j^*$).
This means that $P_i$ can be reconfigured to a curve $P'_i$ as illustrated in Figure~\ref{fig:adjacent}.
This process eliminates at least two intersections and we have $w'_j \in \ang{x_j}$ for any $j \in [k]$ by Lemma~\ref{lem:invariant}.
Thus, the induction hypothesis concludes that $\mathcal{P}'$ is reconfigurable to $\mathcal{Q}$.
\end{proof}

\begin{figure}[t]
    \centering
    \includegraphics{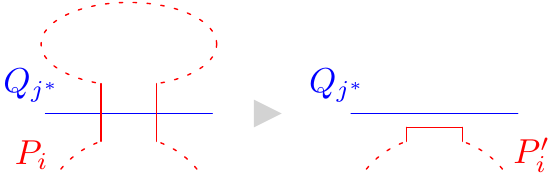}
    \caption{A reconfiguration of $P_i$ to $P'_i$.}
    \label{fig:adjacent}
\end{figure}

By Theorem~\ref{thm:genus0} and Remark~\ref{rem:algintnum}, we obtain the following corollary. 

\begin{corollary}
\label{cor:algintnum}
Let $\mathcal{P} = (P_1, \dots , P_k)$ and $\mathcal{Q} = (Q_1, \dots , Q_k)$ be linkages on a plane \textup{(}or a sphere\textup{)}.
If $\mathcal{P}$ is reconfigurable to $\mathcal{Q}$, then $\mu(P_i, Q_j) = 0$ for any distinct $i, j \in [k]$. 
\end{corollary}

It is worth mentioning that the converse is not necessarily true as illustrated in Figure~\ref{fig:counterexample}.
This means that a ``non-commutative'' tool such as the free group $F_k$ is essential to describe the complexity of the reconfiguration of curves on a plane.

\section{Algorithms for Planar Graphs}
\label{sec:planar_graph}

In this section, we consider the reconfiguration in planar graphs and prove 
Theorems~\ref{thm:2facecalgo}, \ref{thm:stplanaralgo}, and \ref{thm:stplanarchara}. 
We deal with one-face instances and two-face instances of \DPR in Sections~\ref{sec:1face} and \ref{sec:2face}, respectively. 
Then, we discuss \stDPR in Section~\ref{sec:stalgorithm}. 
A proof of a key theorem (Theorem~\ref{thm:2facechara}) is postponed to Section~\ref{sec:proof2face}.

\subsection{One-Face Instance}
\label{sec:1face}

We say that an instance $(G, \mathcal{P}, \mathcal{Q})$ of \DPR is a \emph{one-face instance}
if $G$ is a plane graph and all the terminals are on the boundary of some face. 
We show that $\mathcal{P}$ is always reconfigurable to $\mathcal{Q}$ in a one-face instance. 

\begin{proposition}
\label{prop:1face}
For any one-face instance $(G, \mathcal{P}, \mathcal{Q})$ of \DPR, 
$\mathcal{P}$ is reconfigurable to $\mathcal{Q}$. 
\end{proposition}

\begin{proof}
We prove the proposition by induction on the number of vertices in $G$. 
Let $I = (G, \mathcal{P}, \mathcal{Q})$ be a one-face instance of \DPR. 

If $G$ is not connected, then we can consider each connected component separately. 
If $G$ is connected but not $2$-connected, then there exist subgraphs $G_1$ and $G_2$ such that 
$G = G_1 \cup G_2$ and $V(G_1) \cap V(G_2) = \{v\}$ for some $v \in V$. 
For $i=1, 2$, define $I_i$ as the restriction of $I$ to $G_i$ if some path in $\mathcal{P} \cup \mathcal{Q}$ uses an edge in $G_i$ incident to $v$, 
and define $I_i$ as the restriction of $I$ to $G_i \setminus \{v\}$ otherwise. 
Since $I_1$ and $I_2$ are one-face instances (or trivial instances with at most one terminal pair), the induction hypothesis shows that 
there exist reconfiguration sequences for $I_1$ and $I_2$. 
By combining them, we obtain a reconfiguration sequence from $\mathcal{P}$ to $\mathcal{Q}$.

In what follows, suppose that $G$ is $2$-connected.  
Let $F$ be a face of $G$ whose boundary contains all the terminals. 
Note that the boundary of $F$ forms a cycle, because $G$ is $2$-connected. 
Since there is a linkage, for some $i \in [k]$, there exists an $s_i$-$t_i$ path $R_i$ along $\bd{F}$  
that contains no terminals other than $s_i$ and $t_i$. 
Since $\mathcal{P}$ and $\mathcal{Q}$ are linkages, $P_j$ and $Q_j$ are disjoint from $R_i$ for $j \in [k] \setminus \{i\}$. 
Define $\mathcal{P}'$ (resp.~$\mathcal{Q}'$) as the linkage that is obtained from $\mathcal{P}$ (resp.~$\mathcal{Q}$) by replacing $P_i$ (resp.~$Q_i$) with $R_i$. 
Then, $\mathcal{P} \leftrightarrow \mathcal{P}'$ and $\mathcal{Q} \leftrightarrow \mathcal{Q}'$. 

Since $(G \setminus R_i, \mathcal{P}' \setminus \{R_i\}, \mathcal{Q}' \setminus \{R_i\})$ is a one-face instance, 
$\mathcal{P}' \setminus \{R_i\}$ is reconfigurable to $\mathcal{Q}' \setminus \{R_i\}$ in $G \setminus R_i$ by the induction hypothesis. 
This implies that $\mathcal{P}'$ is reconfigurable to $\mathcal{Q}'$ in $G$. 
Therefore, $\mathcal{P}$ is reconfigurable to $\mathcal{Q}$ in $G$. 
\end{proof}

\subsection{Two-Face Instance}
\label{sec:2face}

Let $k \ge 2$. 
We say that an instance $(G, \mathcal{P}, \mathcal{Q})$ of \DPR is a \emph{two-face instance}
if $G=(V, E)$ is a plane graph, $s_1, \dots , s_k$ are on the boundary of some face $S$, and 
$t_1, \dots , t_k$ are on the boundary of another face $T$.
The objective of this subsection is to present a polynomial-time algorithm for two-face instances.

It suffices to consider the case when the graph is $2$-connected, 
since otherwise we can easily reduce to the $2$-connected case. 
Hence, we may assume that the boundary of each face forms a cycle. 
For ease of explanation, without loss of generality, we assume that $G$ is embedded on $\R^2$ so that $S$ is an inner face and $T$ is the outer face. 
Furthermore, we may assume that
$s_1, \dots , s_k$ lie on the boundary of $S$ clockwise in this order and 
$t_1, \dots , t_k$ lie on the boundary of $T$ clockwise in this order, because 
there is a linkage. 

A vertex set $U \subseteq V$ is called 
a \emph{terminal separator} if 
$U$ separates $\{s_1, \dots , s_k\}$ and $\{t_1, \dots , t_k\}$.  
For two curves (or paths) $P$ and $Q$ between $\bd{S}$ and $\bd{T}$ that share no endpoints, 
define $\mu(P, Q)$ as in Definition~\ref{def:local_int}. 
That is, $\mu(P, Q)$ is the number of times $P$ crosses $Q$ from left to right
minus the number of times $P$ crosses $Q$ from right to left, 
where we suppose that $P$ and $Q$ are oriented from $\bd{S}$ to $\bd{T}$. 
Since $\mu(P_i, Q_j)$ takes the same value for distinct $i, j \in [k]$ (see Appendix~\ref{sec:2facesamevalue}), this value is denoted by $\mu(\mathcal{P}, \mathcal{Q})$. 
Roughly, $\mu(\mathcal{P}, \mathcal{Q})$ indicates the difference of the numbers of rotations around $S$ of the linkages.

The existence of a linkage shows that the graph has no terminal separator of size less than $k$. 
If the graph has no terminal separator of size $k$, then we can characterize the reconfigurability by using $\mu(\mathcal{P}, \mathcal{Q})$. 
The following is a key theorem in our algorithm, whose proof is given in Section~\ref{sec:proof2face}. 

\begin{theorem}
\label{thm:2facechara}
Let $k \ge 2$. 
Suppose that a two-face instance $(G, \mathcal{P}, \mathcal{Q})$ of \DPR has no terminal separator of size $k$. 
Then, $\mathcal{P}$ is reconfigurable to $\mathcal{Q}$ if and only if 
$\mu(\mathcal{P}, \mathcal{Q}) = 0$. 
\end{theorem}

By using this theorem, we can design a polynomial-time algorithm for two-face instances of \DPR and prove Theorem~\ref{thm:2facecalgo}.

\begin{proof}[Proof of Theorem~\ref{thm:2facecalgo}]
Suppose that we are given a two-face instance  $I = (G, \mathcal{P}, \mathcal{Q})$ of \DPR. 

We first test whether $I$ has a terminal separator of size $k$, 
which can be done in polynomial time by a standard minimum cut algorithm.
If there is no terminal separator of size $k$, then Theorem~\ref{thm:2facechara} shows that
we can easily solve \DPR by checking whether $\mu(\mathcal{P}, \mathcal{Q}) = 0$ or not. 

Suppose that we obtain a terminal separator $U$ of size $k$. 
Then, we obtain subgraphs $G_1$ and $G_2$ of $G$ such that $G=G_1 \cup G_2$, $V(G_1) \cap V(G_2) = U$, 
$\{s_1, \dots , s_k\} \subseteq V(G_1)$, and $\{t_1, \dots , t_k\} \subseteq V(G_2)$. 
We test whether $V(P_i) \cap U = V(Q_i) \cap U$ holds for any $i \in [k]$ or not, where we note that each of $V(P_i) \cap U$ and $V(Q_i) \cap U$ consists of a single vertex. 
If this does not hold, then we can immediately conclude that $\mathcal{P}$ is not reconfigurable to $\mathcal{Q}$, 
because $V(P_i) \cap U$ does not change in the reconfiguration. 
If $V(P_i) \cap U = V(Q_i) \cap U$ for $i \in [k]$, then 
we consider the instance $I_i =(G_i, \mathcal{P}_i, \mathcal{Q}_i)$ for $i=1, 2$, 
where $\mathcal{P}_i$ and $\mathcal{Q}_i$ are the restrictions of $\mathcal{P}$ and $\mathcal{Q}$ to $G_i$. 
That is, $I_i$ is the restriction of $I$ to $G_i$. 
Then, we see that $\mathcal{P}$ is reconfigurable to $\mathcal{Q}$ if and only if 
$\mathcal{P}_i$ is reconfigurable to $\mathcal{Q}_i$ for $i=1, 2$. 
Since $I_1$ and $I_2$ are one-face or two-face instances, by solving them recursively, 
we can solve the original instance $I$ in polynomial time. 
See Algorithm~\ref{alg:2face} for a pseudocode of the algorithm. 
\end{proof}

\begin{algorithm}[h]
    \KwInput{A two-face instance $I = (G, \mathcal{P}, \mathcal{Q})$ of DPR.} 
    \KwOutput{Is $\mathcal{P}$ reconfigurable to $\mathcal{Q}$?}
    Compute a terminal separator $U$ of size $k$\; 
    \If{such $U$ does not exist}{
    	\If{$\mu(\mathcal{P}, \mathcal{Q}) = 0$}{\Return{\yes}}
    	\Else{\Return{\no}}
    	}
    \Else{
        \If{$V(P_i) \cap U \neq V(Q_i) \cap U$ for some $i \in [k]$}{\Return{\no}}
        \Else{
        Construct $G_1$ and $G_2$\;
        Solve the restriction $I_i$ of $I$ to $G_i$ for $i=1, 2$, recursively\;
            \If{$I_i$ is a {\yes}-instance for $i=1, 2$}{\Return{\yes}}
            \Else{\Return{\no}}
        }
    }
    \caption{Algorithm for two-face instances of DPR}\label{alg:2face}
\end{algorithm}

\subsection{Reconfiguration of $s$-$t$ Paths}
\label{sec:stalgorithm}

In this subsection, for \stDPR in planar graphs, 
we show results that are analogous to Theorems~\ref{thm:2facechara} and~\ref{thm:2facecalgo}, 
which have been already stated in Section~\ref{sec:contribution}. 

\stcharacterization*

\begin{proof}
Suppose that $G$, $s$, $t$, $\mathcal{P}$, and $\mathcal{Q}$ are as in the statement, and 
assume that there is no $s$-$t$ separator of size $k$. 
We fix an embedding of $G$ on the plane. 
If there is an edge connecting $s$ and $t$, then 
$s$ and $t$ are on the boundary of some face, and hence $\mathcal{P}$ is reconfigurable to $\mathcal{Q}$ in the same way as Proposition~\ref{prop:1face}. 
Thus, it suffices to consider the case when there is no edge connecting $s$ and $t$. 

We now construct an instance of \DPR by replacing $s$ and $t$ with large ``grids'' as follows. 
Let $e_1, e_2, \dots , e_{\ell}$ be the edges incident to $s$ clockwise in this order. 
Note that $\ell \ge k+1$ holds, because $G$ has no $s$-$t$ separator of size $k$. 
For $i \in [\ell]$, we subdivide $e_i$ by introducing $p$ new vertices $v^1_i, v^2_i, \dots , v^{p}_i$ such that 
they are aligned in this order and $v^1_i$ is closest to $s$, where $p$ is a sufficiently large integer (e.g., $p \ge |V|^2)$. 
For $i \in [\ell]$ and for $j \in [p]$, we introduce a new edge connecting $v^j_i$ and $v^j_{i+1}$, where $v^j_{\ell+1} = v^j_{1}$. 
Define $s_i = v^1_i$ for $i \in [k]$ and remove $s$. 
Then, the graph is embedded on the plane and $s_1, \dots , s_k$ are on the boundary of some face clockwise in this order; see \figurename~\ref{fig:410}. 
By applying a similar procedure to $t$, 
we modify the graph around $t$ and define $t_1, \dots , t_k$ that are on the boundary of some face counter-clockwise in this order. 
Let $G'$ be the obtained graph. 
Observe that $G'$ contains no terminal separator of size $k$, because
$G$ has no $s$-$t$ separator of size $k$. 

\begin{figure}[t]
    \centering
    \includegraphics{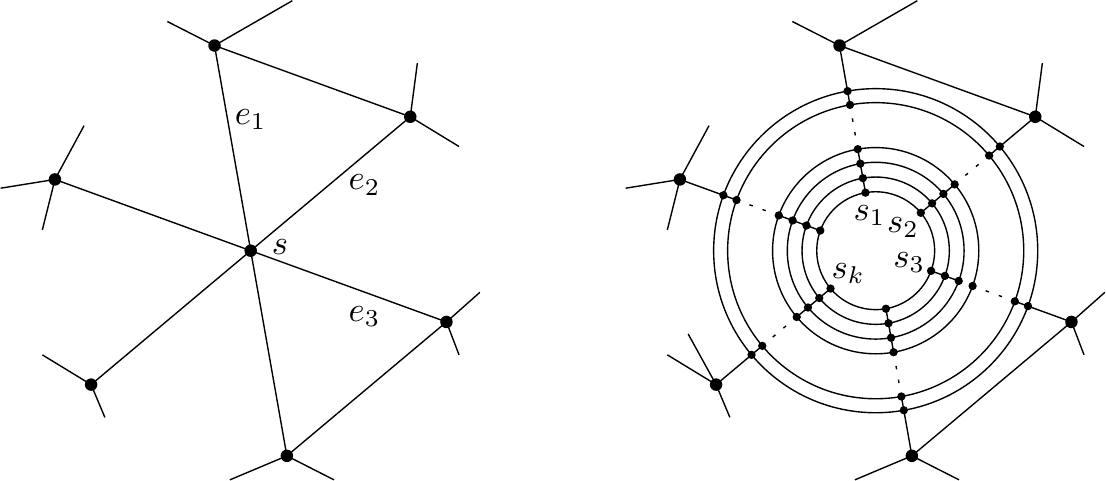}
    \caption{(Left) Original graph $G$.
    (Right) Modification around $s$.}
    \label{fig:410}
\end{figure}

By rerouting the given $s$-$t$ linkages $\mathcal{P}$ and $\mathcal{Q}$ around $s$ and $t$, 
we obtain linkages $\mathcal{P}'$ and $\mathcal{Q}'$ from $\{s_1, \dots , s_k\}$ to $\{t_1, \dots , t_k\}$ in $G'$. 
Note that the restrictions of $\mathcal{P}$ and $\mathcal{Q}$ to $G \setminus \{s, t\}$ coincide with those of $\mathcal{P}'$ and $\mathcal{Q}'$, respectively. Then, we can take $\mathcal{P}'$ and $\mathcal{Q}'$ so that $|\mu(\mathcal{P}', \mathcal{Q}')| \le |V|$. 
Furthermore, by using at most $|V|$ concentric cycles around $s$ and $t$, we can reroute the linkages so that the value $\mu(\mathcal{P}', \mathcal{Q}')$ decreases or increases by one. Therefore, by using $p \ge |V|^2$ concentric cycles, we can reroute $\mathcal{P}'$ and $\mathcal{Q}'$ so that $\mu(\mathcal{P}', \mathcal{Q}')$ becomes zero. 

By Theorem~\ref{thm:2facechara}, $\mathcal{P}'$ is reconfigurable to $\mathcal{Q}'$ in $G'$ (in terms of \DPR). 
Then, the reconfiguration sequence from $\mathcal{P}'$ to $\mathcal{Q}'$ corresponds to that from 
$\mathcal{P}$ to $\mathcal{Q}$ in $G$ (in terms of \stDPR). 
Therefore, $\mathcal{P}$ is reconfigurable to $\mathcal{Q}$ in $G$. 
\end{proof}

\stalgorithm*

\begin{proof}
Suppose that we are given a planar graph $G=(V, E)$ with $s, t \in V$ 
and $s$-$t$ linkages $\mathcal{P}= \{P_1, \dots , P_k\}$ and $\mathcal{Q}= \{Q_1, \dots , Q_k\}$ in $G$. 
We first test whether $G$ has an $s$-$t$ separator of size $k$. 
If there is no such a separator, then 
we can immediately conclude that $\mathcal{P}$ is reconfigurable to $\mathcal{Q}$ by Theorem~\ref{thm:stplanarchara}. 

Suppose that $G$ has an $s$-$t$ separator of size $k$. 
Let $X$ be the inclusionwise minimal vertex set 
subject to $s \in X$ and $N(X)$ is an $s$-$t$ separator of size $k$. 
Note that such $X$ is uniquely determined by the submodularity of $|N(X)|$ and 
it can be computed in polynomial time by a standard minimum cut algorithm. 
Similarly, let $Y$ be the unique inclusionwise minimal vertex set 
subject to $t \in Y$ and $N(Y)$ is an $s$-$t$ separator of size $k$. 
Let $U=N(X)$, $W = N(Y)$, $G_1 = G[X \cup U]$, $G_2 = G \setminus (X \cup Y)$, and $G_3 = G[Y \cup W]$; see \figurename~\ref{fig:413}. 
Since $V(P_i) \cap U$ and $V(P_i) \cap W$ do not change in the reconfiguration, 
we can consider the reconfiguration in $G_1$, $G_2$, and $G_3$, separately. 

\begin{figure}
    \centering
    \includegraphics{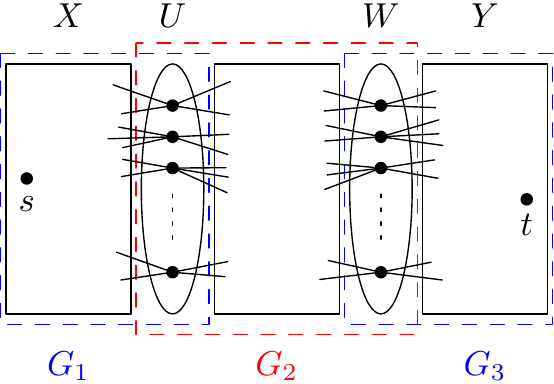}
    \caption{Construction of $G_1$, $G_2$, and $G_3$.}
    \label{fig:413}
\end{figure}

We first consider the reconfiguration in $G_1$. Observe that each path in $\mathcal{P}$ contains exactly one vertex in $U$, and  
the restriction of $\mathcal{P}$ to $G_1$ consists of $k$ paths from $s$ to $U$ 
that are vertex-disjoint except at $s$. The same for $\mathcal{Q}$. 
By the minimality of $X$, $G_1$ contains no vertex set of size $k$ that separates $\{s\}$ and $U$.  
Therefore, by the same argument as Theorem~\ref{thm:stplanarchara}, 
the restriction of $\mathcal{P}$ to $G_1$ is reconfigurable to that of $\mathcal{Q}$. 

If $U \cap W \neq \emptyset$, then $G \setminus X$ contains no vertex set of size $k$ that separates $U$ and $\{t\}$ by the minimality of $Y$.
In such a case, by shrinking $U$ to a single vertex and by applying the same argument as above, 
the restriction of $\mathcal{P}$ to $G \setminus X$ is reconfigurable to that of $\mathcal{Q}$. 
By combining the reconfiguration in $G_1$ and that in $G\setminus X$, we obtain a reconfiguration sequence from $\mathcal{P}$ to $\mathcal{Q}$. 

Therefore, it suffices to 
consider the case when $U \cap W = \emptyset$. 
In the same way as $G_1$, we see that 
the restriction of $\mathcal{P}$ to $G_3$ is reconfigurable to that of $\mathcal{Q}$. 
This shows that the reconfigurability from $\mathcal{P}$ to $\mathcal{Q}$ in $G$ is equivalent to that in $G_2$. 
By changing the indices if necessary, we may assume that $P_i \cap U = Q_i \cap U$ for $i \in [k]$. 
If $P_i \cap W \neq Q_i \cap W$ for some $i \in [k]$, then we can conclude that $\mathcal{P}$ is not reconfigurable to $\mathcal{Q}$. 
Otherwise, let $\mathcal{P}'$ and $\mathcal{Q}'$ be the restrictions of $\mathcal{P}$ and $\mathcal{Q}$ to $G_2$, respectively. 
Since $(G_2, \mathcal{P}', \mathcal{Q}')$ is a one-face or two-face instance of \DPR, we can solve it in polynomial time by Proposition~\ref{prop:1face} and Theorem~\ref{thm:2facecalgo}.  
Therefore, we can test the reconfigurability from $\mathcal{P}$ to $\mathcal{Q}$ in polynomial time;  
see Algorithm~\ref{alg:stplanar}. 
\end{proof}

\begin{algorithm}[t]
    \KwInput{A planar graph $G$ and $s$-$t$ linkages $\mathcal{P}$ and $\mathcal{Q}$.}
    \KwOutput{Is $\mathcal{P}$ reconfigurable to $\mathcal{Q}$?}
    Compute $X$ and $Y$ together with $s$-$t$ separators $U$ and $W$ of size $k$\; 
    \If{such separators do not exist or $U \cap W \neq \emptyset$}{
        \Return{\yes}
    	}
    \Else{
        Construct $G_2 = G \setminus (X\cup Y)$\;
        Change the indices so that $P_i \cap U = Q_i \cap U$ for $i \in [k]$\;
        \If{$P_i \cap W \neq Q_i \cap W$ for some $i \in [k]$}{\Return{\no}}
        \Else{
            Let $\mathcal{P}'$ and $\mathcal{Q}'$ be the restrictions of $\mathcal{P}$ and $\mathcal{Q}$ to $G_2$\;
            \If{$(G_2, \mathcal{P}', \mathcal{Q}')$ is a {\yes}-instance of DPR}{
                \Return{\yes}
            }
            \Else{
                \Return{\no}
            }
        }
    }
    \caption{Algorithm for planar $s$-$t$ paths reconfiguration}\label{alg:stplanar}
\end{algorithm}

\section{Proof of Theorem~\ref{thm:2facechara}}
\label{sec:proof2face}

The necessity (``only if'' part) in Theorem~\ref{thm:2facechara} is immediately derived from Corollary~\ref{cor:algintnum}. 

In what follows in this section, we show the sufficiency (``if'' part) in Theorem~\ref{thm:2facechara}, 
which is one of the main technical contributions in this paper. 
Assume that $\mu(\mathcal{P}, \mathcal{Q}) = 0$ and there is no terminal separator of size $k$. 
The objective is to show that $\mathcal{P}$ is reconfigurable to $\mathcal{Q}$. 
Our proof is constructive, and based on topological arguments. 
A similar technique is used in \cite{DBLP:conf/infocom/KobayashiO14,MCDIARMID1994169,10.5555/645896.758324,DBLP:journals/eor/Otsuki0M16}.

\subsection{Preliminaries for the Proof}

Let $C$ be a simple curve connecting the boundaries of $S$ and $T$ 
such that $C$ contains no vertex in $G$,
$C$ intersects the boundaries of $S$ and $T$ only at its endpoints, and 
$\mu(P_i, C) = 0$ for $i \in [k]$.  
Note that such $C$ always exists, because the last condition is satisfied if $C$ is disjoint from $\mathcal{P}$. 
Note also that $\mu(Q_i, C) = 0$ holds for $i \in [k]$, because $\mu(\mathcal{P}, \mathcal{Q}) = 0$.  

Since $T$ is the outer face, $\R^2 \setminus (S \cup T)$ forms an annulus (or a cylinder).\footnote{More precisely, the annulus is degenerated when $\bd{S} \cap \bd{T} \neq \emptyset$, but the same argument works even for this case.} 
Thus, by cutting it along $C$, we obtain a rectangle 
whose boundary consists of $\bd{S}$, $\bd{T}$, and two copies of $C$.  
We take infinite copies of this rectangle and glue them together to obtain an infinite long strip $R$. 
That is, for $j \in \Z$, let $C^j$ be a copy of $C$,  
let $R^j$ be a copy of the rectangle whose boundary contains $C^j$ and $C^{j+1}$, and 
define $R = \bigcup_{j \in \Z} R^j$; see \figurename~\ref{fig:401}. 
By taking $C$ appropriately, we may assume that the copies of $s_1, \dots , s_k$ lie on the boundary of $R^j$ in this order 
so that $s_1$ is closest to $C^j$ and $s_k$ is closest to $C^{j+1}$. The same for $t_1, \dots , t_k$. 
Note that $R$ is called the \emph{universal cover} of $\R^2 \setminus (S \cup T)$ in the terminology of topology. 

Since $G$ is embedded on $\R^2 \setminus (S \cup T)$, 
this operation naturally defines an infinite periodic graph $\hat{G} = (\hat{V}, \hat{E})$ on $R$ that consists of copies of $G$. 
A path in $\hat G$ is identified with the corresponding curve in $R$. 
For $v \in V$ and $j \in \Z$, 
let $v^j \in \hat{V}$ denote the unique vertex in $R^j$ that corresponds to $v$. 
Since $\mu(P_i, C) = 0$ for $i \in [k]$, 
each path in $\hat{G}$ corresponding to $P_i$ is from 
$s^j_i$ to $t^{j}_i$ for some $j \in \Z$, 
and we denote such a path by $P^j_i$. 
We define $Q^j_i$ in the same way. 
Since $\cal P$ and $\mathcal{Q}$ are linkages in $G$, 
$\{P^j_i \mid i \in [k],\ j \in \Z\}$ and $\{Q^j_i \mid i \in [k],\ j \in \Z\}$ are sets of vertex-disjoint paths in $\hat{G}$.

\begin{figure}[t]
    \centering
    \includegraphics{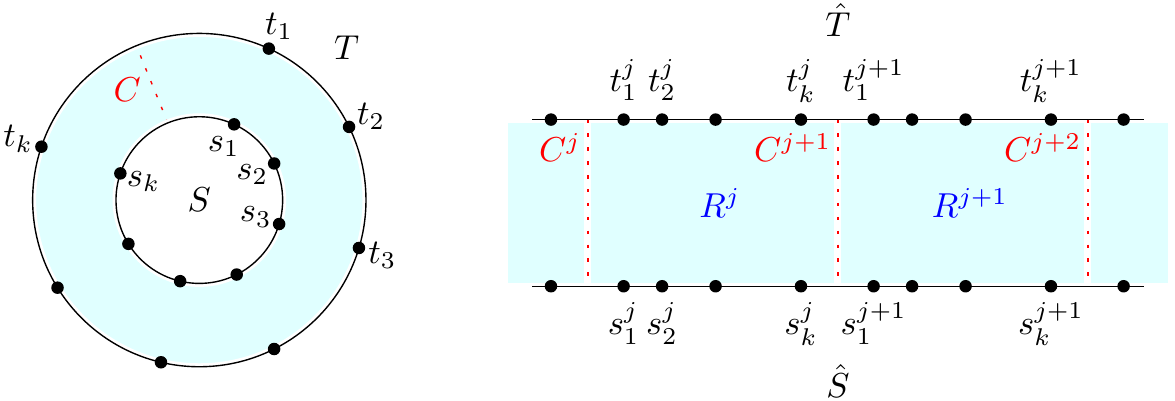}
    \caption{(Left) Curve $C$ in $\R^2 \setminus (S \cup T)$.
    (Right) Construction of $R$.}
    \label{fig:401}
\end{figure}

A path in $\hat{G}$ connecting the boundary of $R$ corresponding to $\bd{S}$ 
and that corresponding to $\bd{T}$ is called an \emph{$\hat{S}$-$\hat{T}$ path}. 
For an $\hat{S}$-$\hat{T}$ path $P$, let $L(P)$ be the region of $R \setminus P$ that is on the ``left-hand side'' of $P$. 
Formally, let $r$ be a point in $R^j$ for sufficiently small $j$, and define $L(P)$ as the set of points $x \in R \setminus P$
such that any curve in $R$ between $r$ and $x$ crosses $P$ an even number of times; see \figurename~\ref{fig:402}. 
For two $\hat{S}$-$\hat{T}$ paths $P$ and $Q$, we denote $P \preceq Q$ if $L(P) \subseteq L(Q)$, and 
denote $P \prec Q$ if $L(P) \subsetneq L(Q)$. 
For two linkages $\mathcal{P} = (P_1, \dots , P_k)$ and $\mathcal{Q} = (Q_1, \dots , Q_k)$ in $G$ with $\mu(P_i, C) = \mu(Q_i, C) = 0$ for $i \in [k]$, 
we denote $\mathcal{P} \preceq \mathcal{Q}$ if $P^j_i \preceq Q^j_i$ for any $i \in [k]$ and $j \in \Z$, and 
denote $\mathcal{P} \prec \mathcal{Q}$ if $\mathcal{P} \preceq \mathcal{Q}$ and $\mathcal{P} \neq \mathcal{Q}$. 

\begin{figure}[t]
    \centering
    \includegraphics{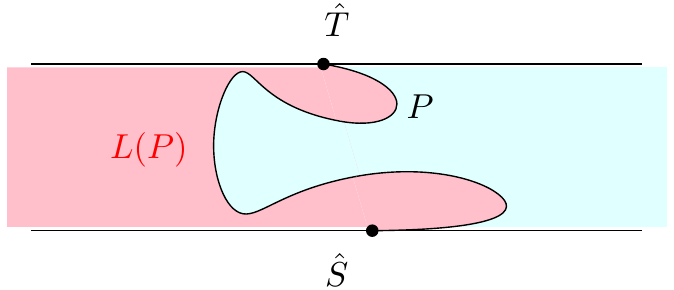}
    \caption{Definition of $L(P)$.}
    \label{fig:402}
\end{figure}

\subsection{Case When $\mathcal{P} \preceq \mathcal{Q}$}
\label{sec:specialcaseplanar}

In this subsection, 
we consider the case when $\mathcal{P} \preceq \mathcal{Q}$, and the general case will be dealt with in Section~\ref{sec:generalcaseplanar}. 
In order to show that $\mathcal{P}$ is reconfigurable to $\mathcal{Q}$, 
we show the following lemma. 

\begin{lemma}
If $\mathcal{P} \prec \mathcal{Q}$, then 
there exists a linkage $\mathcal{P}'$ such that 
$\mathcal{P} \leftrightarrow \mathcal{P}'$ and $\mathcal{P} \prec \mathcal{P}' \preceq \mathcal{Q}$. 
\end{lemma}

\begin{proof}
To derive a contradiction, assume that such $\mathcal{P'}$ does not exist. 
Let $\hat{W} := \{\hat{v} \in \hat{V} \mid \hat{v} \in P^j_i \setminus Q^j_i \mbox{ for some } i \in [k] \mbox{ and } j \in \Z\}$ and 
let $W$ be the subset of $V$ corresponding to $\hat{W}$.  
If $W = \emptyset$, then take an index $i \in [k]$ such that $P_i \neq Q_i$ and 
let $\mathcal{P}' = (P'_1, \dots , P'_k)$ be the set of paths obtained from $\mathcal{P}$ by replacing $P_i$ with $Q_i$. 
Since $\mathcal{Q}$ is a linkage and all the vertices in $P'_{h}$ are contained in $Q_{h}$ for any $h \in [k]$, 
$\mathcal{P}'$ is a desired linkage, 
which contradicts the assumption. 

Thus, it suffices to consider the case when $W \neq \emptyset$. 
Let $u \in W$. 
Let $\hat{u} \in \hat{W}$ be a vertex corresponding to $u$ 
and let $i \in [k]$ and $j \in \Z$ be the indices such that $\hat{u} \in P^j_i \setminus Q^j_i$. 
Since $\hat{u} \in P^j_i \setminus Q^j_i$ implies $\hat{u} \in L(Q^j_i) \setminus L(P^j_i)$, 
there exists a face $\hat{F}$ of $\hat{G}$ such that $\bd{\hat{F}}$ contains an edge of $P^j_i$ incident to $\hat{u}$ and 
$\hat{F} \subseteq L(Q^j_i) \setminus L(P^j_i)$.
Define $(P^j_i)'$ as the $s^j_i$-$t^j_i$ path in $\hat{G}$ 
with maximal $L((P^j_i)')$ subject to $(P^j_i)' \subseteq P^j_i \cup \bd{\hat{F}}$; see \figurename~\ref{fig:403}. 
Note that 
such a path is uniquely determined and 
$P^j_i \prec (P^j_i)' \preceq Q^j_i$. 

\begin{figure}[t]
    \centering
    \includegraphics{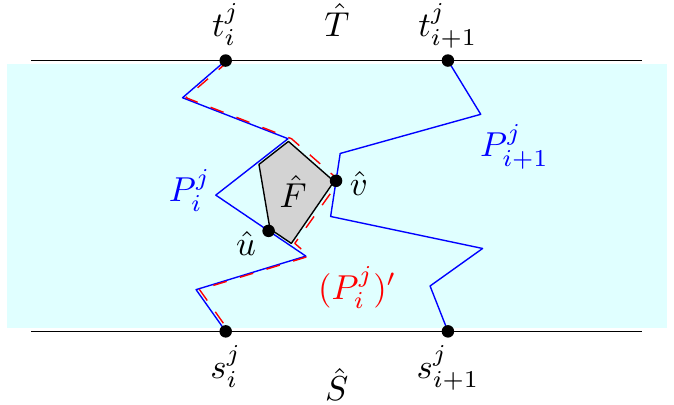}
    \caption{The blue thick paths are $P^j_i$ and $P^{j}_{i+1}$, and the red dashed path is $(P^j_i)'$. There exists a vertex $\hat{v} \in \bd{\hat{F}} \cap P^j_{i+1}$.}
    \label{fig:403}
\end{figure}

Let $P'_i$ be the $s_i$-$t_i$ path in $G$ that corresponds to $(P^j_i)'$. 
If $P'_i$ is disjoint from $P_{h}$ for any $h \in [k] \setminus \{i\}$, then 
we can obtain a desired linkage $\mathcal{P}'$ from $\mathcal{P}$ by replacing $P_i$ with $P'_i$, which contradicts the assumption. 
Therefore, $P'_i$ intersects $P_{h}$ for some $h \in [k] \setminus \{i\}$. 
This together with $P^j_i \prec (P^j_i)'$ shows that $(P^j_i)'$ intersects $P^j_{i+1}$, where $P^j_{k+1}$ means $P^{j+1}_{1}$. 
Since $P^j_{i}$ and $P^j_{i+1}$ are vertex-disjoint, the intersection of $(P^j_i)'$ and $P^j_{i+1}$ is contained in $\bd{\hat{F}}$, which implies that 
$\bd{\hat{F}} \cap P^j_{i+1}$ contains a vertex $\hat{v} \in \hat V$; see \figurename~\ref{fig:403} again. 
Since $\hat{F} \subseteq L(Q^j_i)$, we obtain $\hat{v} \in L(Q^j_i) \cup Q^j_i \subseteq L(Q^j_{i+1})$, and hence $\hat{v} \not\in Q^j_{i+1}$. 
Let $v$ and $F$ be the vertex and the face of $G$ that correspond to $\hat{v}$ and $\hat{F}$, respectively. 
Then, $\hat{v} \in P^j_{i+1} \setminus Q^j_{i+1}$ implies that $\hat{v} \in \hat{W}$ and $v \in W$. 
Let $J$ be a curve in $F$ from $u$ to $v$. 

By the above argument, for any $u \in W$ on $P_i$, there exist a vertex $v \in W$ on $P_{i+1}$ and a curve $J$ from $u$ to $v$ contained in some face of $G$. 
By repeating this argument and by shifting the indices of $P_i$ if necessary, 
we obtain $v_i$ and $J_i$ for $i=1, 2, \ldots$ such that 
$v_i \in W$ is on $P_i$ (where the index is modulo $k$) and $J_i$ is a curve from $v_i$ to $v_{i+1}$ contained in some face.  
We consider the curve $J$ obtained by concatenating $J_1, J_2, \ldots$ in this order. 
Since $|W|$ is finite, this curve visits the same point more than once, and hence it contains a simple closed curve $C^*$. 
Since $C^*$ is simple and visits vertices on $P_i, P_{i+1}, \dots$ in this order, $C^*$ surrounds $S$ exactly once in the clockwise direction; see \figurename~\ref{fig:404}. 
In particular, $C^*$ contains exactly one vertex on $P_i$ for each $i\in [k]$. 
Let $U$ be the set of vertices in $V$ contained in $C^*$. 
Then, $|U|=k$ and $G \setminus U$ has no path between $\{s_1, \dots , s_k\}$ and $\{t_1, \dots , t_k\}$ by the choice of $C^*$.  
Furthermore, $U$ contains no terminals, because $U \subseteq W$ and $W$ contains no terminals. 
Therefore, $U$ is a terminal separator of size $k$, which contradicts the assumption. 
\end{proof}

\begin{figure}[t]
    \centering
    \includegraphics{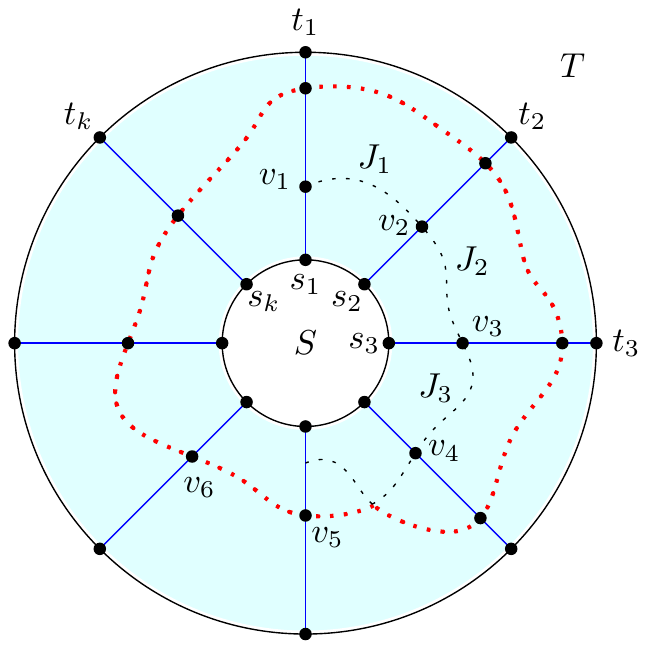}
    \caption{Each blue path represents $P_i$. The dotted curve is part of $J$ and the red dotted thick curve is $C^*$.}
    \label{fig:404}
\end{figure}

As long as $\mathcal{P} \neq \mathcal{Q}$, we apply this lemma and replace $\mathcal{P}$ with $\mathcal{P}'$, repeatedly. 
Then, this procedure terminates when $\mathcal{P} = \mathcal{Q}$, and 
gives a reconfiguration sequence from $\mathcal{P}$ to $\mathcal{Q}$. 
This completes the proof for the case when $\mathcal{P} \preceq \mathcal{Q}$.

\subsection{General Case}
\label{sec:generalcaseplanar}

In this subsection, we consider the case when $\mathcal{P} \preceq \mathcal{Q}$ does not necessarily hold. 
For $i \in [k]$ and $j \in \Z$, 
define $P^j_i \lor Q^j_i$ as the $s^j_i$-$t^j_i$ path in $\hat{G}$ 
with maximal $L(P^j_i \lor Q^j_i)$ subject to $P^j_i \lor Q^j_i \subseteq P^j_i \cup Q^j_i$; see \figurename~\ref{fig:405}. 
Note that 
such a path is uniquely determined,  
$P^j_i \preceq P^j_i \lor Q^j_i$, and $Q^j_i \preceq P^j_i \lor Q^j_i$. 
Since $\hat{G}$ is periodic, for any $j \in \Z$, $P^j_i \lor Q^j_i$ corresponds to a common $s_i$-$t_i$ walk $P_i \lor Q_i$ in $G$. 
We now show that $\mathcal{P} \lor \mathcal{Q} := (P_1 \lor Q_1, \dots , P_k \lor Q_k)$ is a linkage in $G$. 

\begin{figure}[t]
    \centering
    \includegraphics{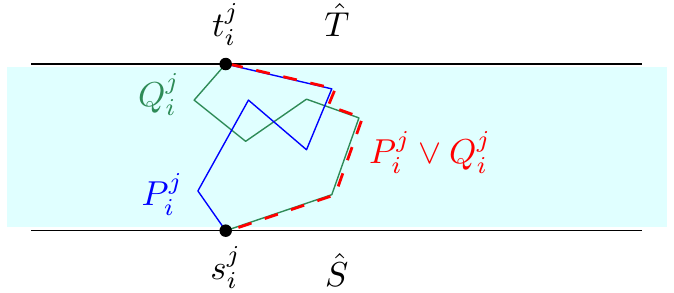}
    \caption{Construction of $P^j_i \lor Q^j_i$.}
    \label{fig:405}
\end{figure}

\begin{lemma}
$\mathcal{P} \lor \mathcal{Q}$ is a linkage in $G$. 
\end{lemma}

\begin{proof}
We first show that $P_i \lor Q_i$ is a path for each $i \in [k]$. 
Assume to the contrary that $P_i \lor Q_i$ visits a vertex $v \in V$ more than once. 
Then, for $j \in \Z$, there exist $j_1, j_2 \in \Z$ with $j_1 < j_2$ such that 
$P^j_i \lor Q^j_i$ contains both $v^{j_1}$ and $v^{j_2}$. 
Since the path $P^j_i \lor Q^j_i$ is contained in the subgraph $P^j_i \cup Q^j_i$, 
without loss of generality, we may assume that $P^j_i$ contains $v^{j_2}$. 
This shows that $v^{j_1} \in L(P^j_i) \subseteq L(P^j_i \lor Q^j_i)$, 
which contradicts that $v^{j_1}$ is contained in $P^j_i \lor Q^j_i$. 

We next show that $P_1 \lor Q_1, \dots , P_k \lor Q_k$ are pairwise vertex-disjoint. 
Assume to the contrary that $P_i \lor Q_i$ and $P_{i'} \lor Q_{i'}$ contain a common vertex $v \in V$ for distinct $i, i' \in [k]$. 
Since $\hat{G}$ is periodic, there exist $j, j' \in \Z$ such that 
$P^j_i \lor Q^j_i$ and $P^{j'}_{i'} \lor Q^{j'}_{i'}$ contain $v^0$ (i.e., the copy of $v$ in $R^0$). 
We may assume that $(j, i)$ is smaller than $(j', i')$ in the lexicographical ordering, that is, 
either $j < j'$ holds or $j=j'$ and $i < i'$ hold. 
Since $P^j_i \lor Q^j_i \subseteq P^j_i \cup Q^j_i$, we may also assume that $v^0 \in P^j_i$ by changing the roles of $P^j_i$ and $Q^j_i$ if necessary. 
Then, we obtain $v^0 \in P^j_i \subseteq L(P^{j'}_{i'}) \subseteq L(P^{j'}_{i'} \lor Q^{j'}_{i'})$, 
which contradicts that $v^0$ is contained in $P^{j'}_{i'} \lor Q^{j'}_{i'}$. 
\end{proof}

We also see that $\mu(P_i \lor Q_i, C) = 0$ for $i \in [k]$ by definition, and hence $\mu(\mathcal{P}, \mathcal{P} \lor \mathcal{Q})=0$. 
Since $\mathcal{P} \preceq \mathcal{P} \lor \mathcal{Q}$ and $\mu(\mathcal{P}, \mathcal{P} \lor \mathcal{Q})=0$, 
$\cal P$ is reconfigurable to $\mathcal{P} \lor \mathcal{Q}$
as described in Section~\ref{sec:specialcaseplanar}. 
Similarly, $\cal Q$ is reconfigurable to $\mathcal{P} \lor \mathcal{Q}$, 
which implies that $\mathcal{P} \lor \mathcal{Q}$ is reconfigurable to $\cal Q$.  
By combining them, we see that $\cal P$ is reconfigurable to $\cal Q$, 
which completes the proof of the sufficiency in Theorem~\ref{thm:2facechara}.


\section{$\PSPACE$-Completeness}
\label{sec:pspacecompl}

We first observe that \DPR and \stDPR are in $\PSPACE$ by using $\PSPACE = \mathsf{NPSPACE}$ (a corollary of Savitch's theorem \cite{DBLP:journals/jcss/Savitch70}). As a certificate, we receive a reconfiguration sequence. In our polynomial-space algorithm, we read linkages in the certificate one by one, and check whether the linkage is indeed a linkage and whether two consecutive linkages are obtained by a single reconfiguration step. At each step, the working memory only requires to store the graph $G$ and two consecutive linkages. Thus, the algorithm runs in polynomial space. This is a non-deterministic algorithm, but by $\PSPACE = \mathsf{NPSPACE}$, we conclude that the problems belong to $\PSPACE$.

For the proof of $\PSPACE$-hardness, we reduce the \emph{nondeterministic constraint logic reconfiguration} (NCL reconfiguration) to our problem.
In the NCL reconfiguration, we consider an undirected cubic graph, where each edge is associated with weight one or two and each vertex is incident to three weight-$2$ edges (called an \emph{OR vertex}) or incident to one weight-$2$ edge and two weight-$1$ edges (called an \emph{AND vertex}).
We call such an undirected graph an \emph{AND/OR graph}.
An \emph{NCL configuration} is an orientation of the edges of an AND/OR graph in which at each vertex the total weight of the incoming arcs is at least two.
A \emph{flip} of an arc is a process of changing the direction of the arc.
In the NCL reconfiguration, we are given two NCL configurations of an AND/OR graph, and need to determine whether we can obtain one from the other by a sequence of flips in such a way that all the intermediate orientations are NCL configurations.
It is known that the NCL reconfiguration is $\PSPACE$-complete \cite{HD}, even when the AND/OR graph is planar and of bounded bandwidth \cite{DBLP:conf/iwpec/Zanden15}.

\subsection{General Graphs and Two Paths}

\hardness*

\begin{proof}
We have already observed that the problem is in $\PSPACE$. To show the $\PSPACE$-hardness, we now give a transformation of a given AND/OR graph $H=(V(H),E(H))$ to an undirected graph $G=(V(G),E(G))$.
Each weight-$2$ edge $e = \{u, v\} \in E(H)$ of $H$ is mapped to the following \emph{weight-$2$ edge gadget} $G_e$ (see \figurename~\ref{fig:pspacecompl_edge1}):
\begin{align*}
    V(G_e) &= \{w_{e,u}, w_{e,v}, r_{e,1}, r_{e,2}\},\\
    E(G_e) &= \{\{w_{e,u}, r_{e,1}\}, \{w_{e,u}, r_{e,2}\}, \{w_{e,v}, r_{e,1}\}, \{w_{e,v}, r_{e,2}\}\}.
\end{align*}
Similarly, each weight-$1$ edge $e = \{u, v\} \in E(H)$ of $H$ is mapped to the following \emph{weight-$1$ edge gadget} $G_e$ (see \figurename~\ref{fig:pspacecompl_edge1}):
\begin{align*}
    V(G_e) &= \{w_{e,u}, w_{e,v}, b_{e,1}, b_{e,2}\},\\
    E(G_e) &= \{\{w_{e,u}, b_{e,1}\}, \{w_{e,u}, b_{e,2}\}, \{w_{e,v}, b_{e,1}\}, \{w_{e,v}, b_{e,2}\}\}.
\end{align*}

\begin{figure}[t]
    \centering
    \includegraphics{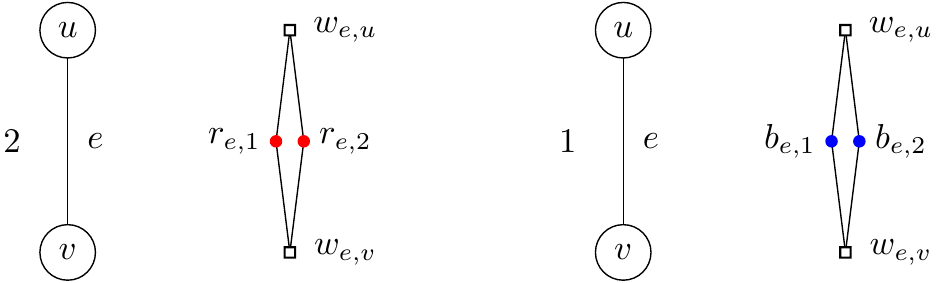}
    \caption{(Left) A weight-$2$ edge gadget.
    (Right) A weight-$1$ edge gadget.}
    \label{fig:pspacecompl_edge1}
\end{figure}

Let $v \in V(H)$ be an OR vertex, and $e,f,g \in E(H)$ be three weight-$2$ edges incident to $v$.
Then, $v$ is mapped to the following \emph{OR vertex gadget} $G_v$ (see \figurename~\ref{fig:pspacecompl_vert1}):
\begin{align*}
    V(G_v) &= \{w_{e,v}, w_{f,v}, w_{g,v}, b_{v,1}, b_{v,2}\},\\
    E(G_v) &= \{\{w_{e,v}, b_{v,1}\}, \{w_{e,v}, b_{v,2}\}, \{w_{f,v}, b_{v,1}\}, \{w_{f,v}, b_{v,2}\},\\
        &\qquad \{w_{g,v}, b_{v,1}\}, \{w_{g,v}, b_{v,2}\}\}.
\end{align*}
Note that the vertices with the same label are identified.

\begin{figure}[t]
    \centering
    \includegraphics{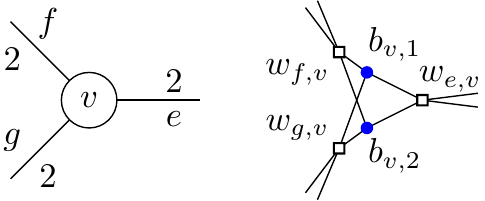}
    \qquad
    \includegraphics{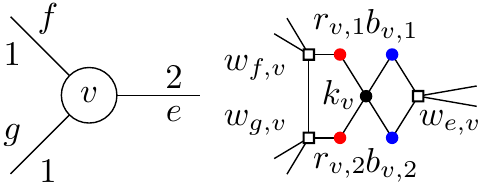}
    \caption{(Left) An OR vertex gadget. (Right) An AND vertex gadget.}
    \label{fig:pspacecompl_vert1}
\end{figure}

An AND vertex $v \in V(H)$ that is incident to one weight-$2$ edge $e$ and two weight-$1$ edges $f, g$ is mapped to the following \emph{AND vertex gadget} $G_v$  (see \figurename~\ref{fig:pspacecompl_vert1}):
\begin{align*}
    V(G_v) &= \{w_{e,v}, w_{f,v}, w_{g,v}, r_{v,1}, r_{v,2}, b_{v,1}, b_{v,2}, k_v\},\\
    E(G_v) &= \{ \{w_{e,v}, b_{v,1}\},
                 \{w_{e,v}, b_{v,2}\},
                 \{w_{f,v}, r_{v,1}\},
                 \{w_{g,v}, r_{v,2}\},\\
           & \qquad
                 \{r_{v,1}, k_v\},
                 \{r_{v,2}, k_v\},
                 \{b_{v,1}, k_v\},
                 \{b_{v,2}, k_v\}, 
                 \{w_{f,v}, w_{g,v}\}\}.
\end{align*}

Define
\[
    V(G) = \{s,t\} \cup \bigcup_{v \in V(H)} V(G_v) \cup \bigcup_{e \in E(H)} V(G_e).
\]

A vertex of $V(G)$ is \emph{red} if it is of the form $r_{x,i}$ for some vertex/edge $x \in V(H) \cup E(H)$ and $i \in \{1,2\}$;
a vertex of $V(G)$ is \emph{blue} if it is of the form $b_{x,i}$ for some vertex/edge $x \in V(H) \cup E(H)$ and $i \in \{1,2\}$;
a vertex of $V(G)$ is \emph{white} if it is of the form $w_{e,v}$ for some vertex $v \in V(H)$ and edge $e \in E(H)$;
a vertex of $V(G)$ is \emph{black} if it is $s$, $t$ or of the form $k_v$ for some vertex $v \in V(H)$.
The number of red vertices is equal to $2|V^{\mathrm{a}}(H)|+2|E_2(H)|$, where $V^{\mathrm{a}}(H)$ is the set of AND vertices and $E_2(H)$ is the set of weight-$2$ edges;
The number of blue vertices is equal to $2|V(H)|+2|E_1(H)|$, where $E_1(H)$ is the set of weight-$1$ edges;
The number of white vertices is equal to $3|V(H)|=2|E(H)|$;
The number of black vertices is equal to $2+|V^{\mathrm{a}}(H)|$.

Assume that the elements of $V^{\mathrm{a}}(H) \cup E_2(H)$ are arbitrarily ordered as $x_1, x_2, \dots, x_R$
and the elements of $V(H) \cup E_1(H)$ are arbitrarily ordered as $y_1, y_2, \dots, y_B$.
Then, define
\begin{align*}
    E(G) &= \bigcup_{v \in V(H)}E(G_v) \cup \bigcup_{e \in E(H)}E(G_e)\\
        &\qquad \cup \{ \{r_{x_i,2}, r_{x_{i+1},1}\} \mid i \in \{1,2,\dots,R-1\}\}\\ 
        &\qquad \cup \{\{b_{y_i,2}, b_{y_{i+1},1}\} \mid i \in \{1,2,\dots,B-1\}\}\\
        &\qquad \cup \{ \{s,r_{x_1,1}\}, \{s, b_{y_1,1}\} \}  \cup \{\{r_{x_R,2}, t\}, \{b_{y_B,2}, t\}\}.
\end{align*}
This completes the construction of $G$.
\figurename~\ref{fig:pspacecompl1} may help the reader understand.

We now describe the way how an NCL configuration of $H$ translates to an $s$-$t$ linkage $\{P_1, P_2\}$ of $G$.
We force that $P_1$ goes through all the red vertices and $P_2$ goes through all the blue vertices.
As a sequence of vertices, $P_1$ is constructed as follows.
\begin{enumerate}
    \item The path starts at $s$. Let $r_{x_0,2} = s$.
    \item Now we iterate over $i\in \{1,\dots,R\}$.
    As an invariant at the beginning of the iteration, we assume that we have just visited the vertex $r_{x_{i-1},2}$. Then, the path continues to $r_{x_i,1}$. We consider two separate cases.
    \begin{enumerate}
    \item Consider the case where $x_i=\{u,v\}$ is a weight-$2$ edge.
    Let the edge be oriented from $u$ to $v$ in the NCL configuration.
    Then, the path visits $w_{x_i,u}$, and continues to $r_{x_i,2}$.
    \item Consider the case where $x_i$ is an AND vertex that is incident to one weight-$2$ edge $e$ and two weight-$1$ edges $f$ and $g$.
    If $e$ is oriented toward $x_i$, then the path visits $k_{x_i}$ and continues to $r_{x_i,2}$.
    If $e$ is not oriented toward $x_i$, then the path visits $w_{f, x_i}, w_{g, x_i}$ and continues to $r_{x_i, 2}$.
    \end{enumerate}
    \item After the whole iteration, the path now visits $r_{x_R, 2}$.
    Then, the path continues to $t$ and we finish the construction of $P_1$.
\end{enumerate}
Similarly, as a sequence of vertices, $P_2$ is constructed as follows.
\begin{enumerate}
    \item The path starts at $s$. Let $b_{y_0,2} = s$.
    \item Now we iterate over $i\in \{1,\dots,B\}$.
    As an invariant at the beginning of the iteration, we assume that we have just visited the vertex $b_{y_{i-1},2}$. Then, the path continues to $b_{y_i,1}$. We consider three separate cases.
    \begin{enumerate}
    \item Consider the case where $y_i=\{u,v\}$ is a weight-$1$ edge.
    Let the edge be oriented from $u$ to $v$ in the NCL configuration.
    Then, the path visits $w_{y_i,u}$, and continues to $b_{y_i,2}$.
    \item Consider the case where $y_i$ is an OR vertex that is incident to three weight-$2$ edges $e$, $f$, and $g$.
    At least one of those edges is directed toward $y_i$.
    Let $e$ be such an edge.
    Then, the path visits $w_{e,y_i}$ and continues to $b_{y_i, 2}$.
    \item Consider the case where $y_i$ is an AND vertex that is incident to one weight-$2$ edge $e$ and two weight-$1$ edges $f$ and $g$.
    If $e$ is oriented toward $y_i$, then the path visits $w_{e,y_i}$ and continues to $b_{y_i,2}$.
    If $e$ is not oriented toward $y_i$, then the path visits $k_{y_i}$ and continues to $b_{y_i, 2}$.
    \end{enumerate}
    \item After the whole iteration, the path now visits $b_{y_B, 2}$.
    Then, the path continues to $t$ and we finish the construction of $P_2$.
\end{enumerate}

The construction is illustrated in \figurename~\ref{fig:pspacecompl1}.
The constructed $s$-$t$ linkage $\mathcal{P}=\{P_1, P_2\}$ is \emph{created} from the NCL configuration $\sigma$.
The following claim ensures that the constructed paths are indeed vertex-disjoint.

\begin{figure}[t]
    \centering
    \includegraphics{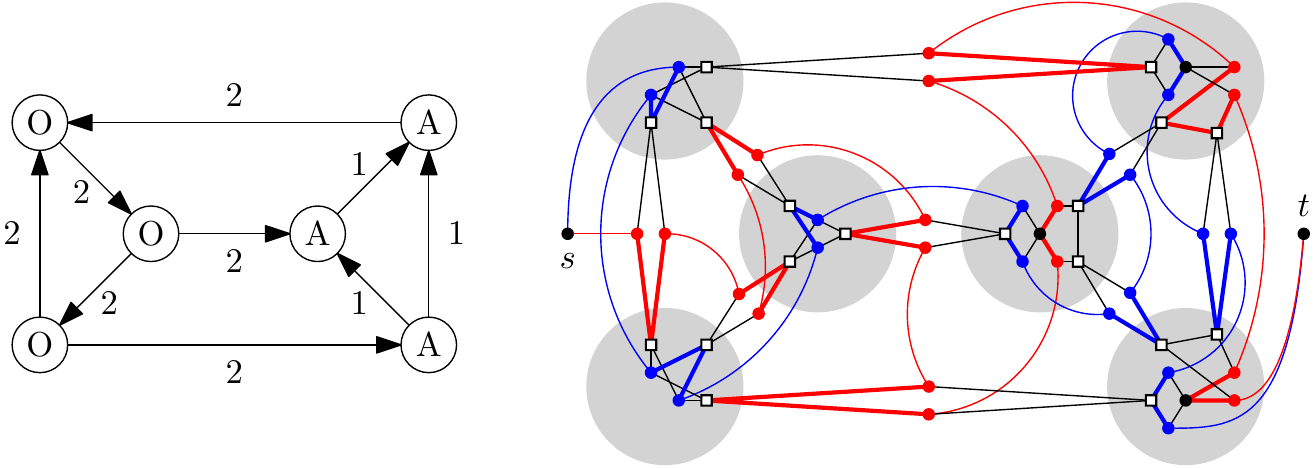}
    \caption{Mapping an NCL configuration to an $s$-$t$ linkage $\{P_1, P_2\}$. In the NCL configuration, the vertex label A stands for an AND vertex, and the vertex label O stands for an OR vertex. The path $P_1$ is colored red and the path $P_2$ is colored blue.}
    \label{fig:pspacecompl1}
\end{figure}

\begin{claim}
\label{clm:pspacecompl-twopaths-vdisj}
The paths $P_1$ and $P_2$ are internally vertex-disjoint.
\end{claim}
\begin{proof}
The paths $P_1$ and $P_2$ start at $s$ and ends at $t$.
By the construction, the red vertices are visited by $P_1$ only and the blue vertices are visited by $P_2$ only.
Hence, if they are not vertex-disjoint, they will share a white vertex or a black vertex.

Suppose that they share a white vertex.
Let one of the shared white vertices be of the form $w_{e,v}$ for some edge $e \in E(H)$ and an end-vertex $v$ of $e$.

As the first case, suppose that $v$ is an OR vertex.
Note that $e$ is a weight-$2$ edge in this case.
Since $w_{e,v}$ is visited by $P_1$, by Step 2(a), $e$ is oriented from $v$ to the other end-vertex of $e$.
Since $w_{e,v}$ is visited by $P_2$, by Step 2(b), $e$ is oriented toward $v$.
Thus, the orientation of $e$ is in conflict and we reach a contradiction.

As the second case, suppose that $v$ is an AND vertex.
We further distinguish two cases.
First, consider the case where $e$ is a weight-$2$ edge.
Then, since $w_{e,v}$ is visited by $P_1$, by Step 2(a), $e$ is oriented from $v$ to the other end-vertex of $e$.
On the other hand, since $w_{e,v}$ is visited by $P_2$, by Step 2(c), $e$ is oriented toward $v$.
Thus, the orientation of $e$ is in conflict and we reach a contradiction.

Second, consider the case where $e$ is a weight-$1$ edge.
Then, since $w_{e,v}$ is visited by $P_2$, by Step 2(a), $e$ is oriented from $v$.
On the other hand, since $w_{e,v}$ is visited by $P_1$, by Step 2(b), the weight-$2$ edge incident to $v$ is not oriented toward $v$.
They together mean that the total weight of incoming arc to $v$ is at most one.
This is a contradiction since the orientation does not meet the requirement to be an NCL configuration.

Next, suppose that $P_1$ and $P_2$ share a black vertex.
Such a black vertex only exists in an AND vertex gadget $G_v$, where $v$ is an AND vertex of $H$ that is incident to one weight-$2$ edge $e$ and two weight-$1$ edges $f$ and $g$.
Namely, the black vertex is of the form $k_v$.
Since $k_v$ is visited by $P_1$, by Step 2(b), $e$ is oriented toward $v$.
On the other hand, since $k_v$ is visited by $P_2$, by Step 2(c), $e$ is not oriented toward $v$.
This is a contradiction.

In all cases, we derive a contradiction. Thus, $P_1$ and $P_2$ are internally vertex-disjoint.
\end{proof}

We now study how a flip in an orientation in an NCL configuration corresponds to a constant-length sequence of reconfiguration steps of paths.
\begin{claim}
\label{clm:pspacecompl-twopaths-NCL-to-disjpaths}
Let $\mathcal{P}=\{P_1, P_2\}$ be an $s$-$t$ linkage of $G$ that is created from an NCL configuration $\sigma$.
If an NCL configuration $\sigma'$ is obtained from $\sigma$ by a flip of a single edge, then an $s$-$t$ linkage $\mathcal{P}'=\{P'_1, P'_2\}$ of $G$ that is created from $\sigma'$ can be obtained from $\mathcal{P}$ by at most a constant number of reconfiguration steps in $G$.
\end{claim}
\begin{proof}
Assume that we are given an NCL configuration $\sigma$ of $H$.
Let $\mathcal{P} = \{P_1, P_2\}$ be an $s$-$t$ linkage of $G$ that is created from $\sigma$.
By the construction, the path $P_1$ passes through all the red vertices; the path $P_2$ passes through all the blue vertices.

As the first case, consider a flip of a weight-$1$ edge $f$ that connects two AND vertices $u,v$ from $(u,v)$ to $(v,u)$.
Then, the blue path $P_2$ passes through $b_{f,1}, w_{f,u}, b_{f,2}$ of the weight-$1$ edge gadget in this order.
Let $P'_2$ be a blue path that is created from the NCL configuration obtained after flipping $(u,v$) to $(v,u)$.
To obtain an $s$-$t$ linkage $\{P_1, P'_2\}$, we replace the part $b_{f,1}, w_{f,u}, b_{f,2}$ of $P_2$ with $b_{f,1}, w_{f,v}, b_{f,2}$.

This is possible only if $P_1$ does not pass through $w_{f,v}$; we now observe this is indeed the case.
If $P_1$ passes through $w_{f,v}$, then by the construction of $P_1$, the path $P_1$ passes through $r_{v,1}, w_{f,v}, w_{g,v}, r_{v,2}$ in this order where $g$ is the other weight-$1$ edge incident to $v$. This happens only when a unique weight-$2$ edge incident to $v$ is not oriented toward $v$. Therefore, after flipping to $(v,u)$ the total sum of the incoming edges to $v$ will be $1$. This is a contradiction to the property of an NCL configuration.

As the second case, consider a flip of a weight-$2$ edge $e$ that connects a vertex $u$ and an AND vertex $v$ from $(u,v)$ to $(v,u)$.
Then, the red path $P_1$ passes through $r_{e,1}, w_{e,u}, r_{e,2}$ of the weight-$2$ edge gadget in this order.
Let $P'_1$ be a red path that is created from the orientation obtained after flipping to $(v,u)$.
To obtain a family of internally vertex-disjoint paths $P'_1, P_2$, we want to replace the part $r_{e,1}, w_{e,u}, r_{e,2}$ of $P_1$ with $r_{e,1}, w_{e,v}, r_{e,2}$.

This is possible only if $P_2$ does not pass through $w_{e,v}$.
However, since $e$ is oriented toward $v$, $P_2$ passes through $w_{e,v}$. 
Thus, in order to make $w_{e,v}$ vacant, we apply the following operations before the above transformation. 
We first replace the part $r_{v,1}$, $k_v$, $r_{v,2}$ of $P_1$ by $r_{v,1}$, $w_{f,v}$, $w_{g,v}$, $r_{v,2}$, where $f$ and $g$ are weight-$1$ edges incident to $v$.
This is possible since $f$ and $g$ are oriented toward $v$ in $\sigma$.
Next, we replace the part $b_{v,1}$, $w_{e,v}$, $b_{v,2}$ of $P_2$ by $b_{v,1}$, $k_v$, $b_{v,2}$.
This way, we can make $w_{e,v}$ vacant with two extra reconfiguration steps.

As the third but final case, consider a flip of a weight-$2$ edge $e$ that connects a vertex $u$ and an OR vertex $v$ from $(u,v)$ to $(v,u)$.
Then, the red path $P_1$ passes through $r_{e,1}, w_{e,u}, r_{e,2}$ of the weight-$2$ edge gadget in this order.
Let $P'_1$ be a red path that is created from the orientation obtained after flipping to $(v,u)$.
To obtain a family of internally vertex-disjoint paths $P'_1, P_2$, we want to replace the part $r_{e,1}, w_{e,u}, r_{e,2}$ of $P_1$ with $r_{e,1}, w_{e,v}, r_{e,2}$.

This is possible only if $P_2$ does not pass through $w_{e,v}$.
If this is the case, there is no problem.
However, in some cases, $P_2$ indeed passes through $w_{e,v}$.
Then, we first replace the part $b_{v,1}, w_{e,v}, b_{v,2}$ with $b_{v,1}, w_{f,v}, b_{v,2}$ or $b_{v,1}, w_{g,v}, b_{v,2}$, where $f,g$ are other weight-$2$ edges incident to $v$, to make $w_{e,v}$ vacant.
We now observe that this extra step can always be done.
Suppose not.
Then, all the three vertices $w_{e,v}, w_{f,v}, w_{g,v}$ are occupied by $P_2$.
This means that the edges $e, f, g$ are oriented as they leave $v$, and the total sum of edges incoming to $v$ is zero.
This is a contradiction to the property of an NCL configuration.

In the second and third cases above, if $u$ is an AND vertex, then 
we replace the part $b_{u,1}$, $k_u$, $b_{u,2}$ of $P_2$ by $b_{u,1}$, $w_{e,u}$, $b_{u,2}$, and 
replace the part $r_{u,1}$, $w_{f',u}$, $w_{g',u}$, $r_{u,2}$ of $P_1$ by $r_{u,1}$, $k_u$, $r_{u,2}$, where $f',g'$ are other weight-$2$ edges incident to $u$. 
Then, we obtain an $s$-$t$ linkage of $G$ that is created from $\sigma'$.  
\end{proof}

To argue the reverse direction of correspondence, we first observe that a reconfiguration step maintains the property that the paths go through the red and blue vertices.

\begin{claim}
\label{clm:pspacecompl-twopaths-onestep}
Let $\mathcal{P}=\{P_1, P_2\}$ be an $s$-$t$ linkage of $G$ such that $P_1$ goes through all the red vertices and $P_2$ goes through all the blue vertices.
If an $s$-$t$ linkage $\mathcal{P}'=\{P'_1, P'_2\}$ of $G$ is obtained by a single reconfiguration step in $G$ from $\mathcal{P}$, then $P'_1$ goes through all the red vertices and $P'_2$ goes through all the blue vertices.
\end{claim}
\begin{proof}
A single reconfiguration step changes one of $P_1$ and $P_2$, but not both. We distinguish two cases.

First, assume that $P_1$ is reconfigured to another path $P'_1$, where in this case $P_2 = P'_2$.
Since all the blue vertices are occupied by $P_2$, the path $P'_1$ cannot pass through any blue vertex.
The removal of the blue vertices from $G$ results in a graph with the following property: each red vertex separates $s$ and $t$.
Therefore, $P'_1$ must go through all the red vertices.

Second, assume that $P_2$ is reconfigured to another path $P'_2$, where in this case $P_1 = P'_1$.
Since all the red vertices are occupied by $P_1$, the path $P'_2$ cannot pass through any red vertex.
The removal of the red vertices from $G$ yields a graph with the following property: each blue vertex separates $s$ and $t$.
Therefore, $P'_2$ must go through all the blue vertices.
\end{proof}

We now describe a canonical way of constructing an NCL configuration $\sigma$ of $H$ from an $s$-$t$ linkage $\mathcal{P}=\{P_1, P_2\}$ of $G$ where $P_1$ goes through all the red vertices and $P_2$ goes through all the blue vertices.

Let $e=\{u,v\}$ be a weight-$2$ edge of $H$.
The weight-$2$ edge gadget $G_e$ of $e$ has vertices $r_{e,1}$ and $r_{e,2}$ that are visited by $P_1$.
The path $P_1$ should visit one of $w_{e,u}$ and $w_{e,v}$, but not both.
If $P_1$ visits $w_{e,u}$, then the edge $e$ is oriented toward $v$ in $\sigma$.
If $P_1$ visits $w_{e,v}$, then the edge $e$ is oriented toward $u$ in $\sigma$.

Let $e=\{u,v\}$ be a weight-$1$ edge of $H$.
The weight-$1$ edge gadget $G_e$ of $e$ has vertices $b_{e,1}$ and $b_{e,2}$ that are visited by $P_2$.
The path $P_2$ should visit one of $w_{e,u}$ and $w_{e,v}$, but not both.
If $P_2$ visits $w_{e,u}$, then the edge $e$ is oriented toward $v$ in $\sigma$.
If $P_2$ visits $w_{e,v}$, then the edge $e$ is oriented toward $u$ in $\sigma$.

This finishes the description of the construction of an NCL configuration.
We call the constructed NCL configuration the \emph{canonical NCL configuration obtained from $\{P_1, P_2\}$}.
We observe that if $\mathcal{P}=\{P_1, P_2\}$ is an $s$-$t$ linkage of $G$ that is created from an NCL configuration $\sigma$, then 
$\sigma$ is the canonical NCL configuration obtained from $\mathcal{P}$.

\begin{claim}
Let $\mathcal{P}=\{P_1, P_2\}$ be an $s$-$t$ linkage of $G$ such that $P_1$ goes through all the red vertices and $P_2$ goes through all the blue vertices.
Then, the canonical NCL configuration $\sigma$ obtained from $\mathcal{P}$ is indeed an NCL configuration of $H$.
\end{claim}
\begin{proof}
We check that $\sigma$ satisfies the weight requirement at every vertex $v$ of $H$.

Assume that $v$ is an OR vertex which is incident to three weight-$2$ edges $e$, $f$, and $g$.
We observe that in the OR vertex gadget $G_v$, one of $w_{e,v}$, $w_{f,v}$ and $w_{g,v}$ is not occupied by the red path $P_1$.
Suppose not.
Then, out of four neighbors of the blue vertex $b_{v,1}$ in $G_v$, the three vertices $w_{e,v}$, $w_{f,v}$ and $w_{g,v}$ are already occupied.
Therefore, the blue path $P_2$ cannot be vertex-disjoint with $P_1$. This is a contradiction.

Therefore, in the canonical NCL configuration $\sigma$, one of the edges $e$, $f$, and $g$ is oriented toward $v$.
This means the incoming weight to $v$ is at least two, and the weight requirement is satisfied at $v$.

Assume that $v$ is an AND vertex which is incident to one weight-$2$ edge $e$ and two weight-$1$ edges $f$ and $g$.
We observe that in the AND vertex gadget $G_v$, if $w_{e,v}$ is occupied by $P_1$, then $w_{f,v}$ and $w_{g,v}$ are occupied by $P_1$, too.
To this end, assume that $w_{e,v}$ is occupied by $P_1$.
Since $P_1$ and $P_2$ are vertex-disjoint, and $b_{v,1}$ and $b_{v,2}$ are of degree $3$ that share $w_{e,v}$ as one of the common neighbors, the blue path $P_2$ must go through the black vertex $k_v$.
Since the red vertices $r_{v,1}$ and $r_{v,2}$ are of degree $3$ that share $k_v$ as a common neighbor, the red path $P_1$ must go through $w_{f,v}$ and $w_{g,v}$.
This finishes the proof of the observation.

Therefore, if $w_{e,v}$ is not occupied by $P_1$, then in the canonical NCL configuration $\sigma$ the weight-$2$ edge $e$ is oriented toward $v$, in which case the weight requirement is satisfied at $v$.
If not, by the observation above, $w_{f,v}$ and $w_{g,v}$ are occupied by $P_1$.
Since $P_1$ and $P_2$ are vertex-disjoint, the blue path $P_2$ cannot go through any of $w_{f,v}$ and $w_{g,v}$.
This means that in the canonical NCL configuration $\sigma$ the two weight-$1$ edges $f$ and $g$ are oriented toward $v$, and in this case the weight requirement is satisfied at $v$.
\end{proof}

\begin{claim}
\label{clm:pspacecompl-twopaths-disjpaths-to-NCL}
Let $\mathcal{P}=\{P_1, P_2\}$ be an $s$-$t$ linkage of $G$ such that $P_1$ goes through all the red vertices and $P_2$ goes through all the blue vertices, and 
let $\sigma$ be the canonical NCL configuration obtained from $\mathcal{P}$. 
Let $\mathcal{P}'=\{P'_1, P'_2\}$ be an $s$-$t$ linkage of $G$ obtained by a single reconfiguration step from $\mathcal{P}$ and $\sigma'$  be the canonical NCL configuration obtained from $\mathcal{P}'$.
Then, $\sigma'$ can be obtained from $\sigma$ by at most $|V(H)|+|E(H)|$ flips in $H$.
\end{claim}
\begin{proof}
Suppose that we now want to transform $P_1$ to $P'_1$ in such a way that $P'_1$ and $P_2$ are internally vertex-disjoint $s$-$t$ paths.
By Claim \ref{clm:pspacecompl-twopaths-onestep}, $P'_1$ goes through all the red vertices. 

Consider the red vertices $r_{e,1}$ and $r_{e,2}$ for an edge $e=\{u,v\} \in E(H)$.
Since $P_1$ passes through both of them, $P_1$ must pass through exactly one of $w_{e,u}$ and $w_{e,v}$.
The same observation applies to $P'_1$, too.
If $P_1$ passes through $w_{e,u}$ and $P_2$ does not pass through $w_{e,v}$, then we may reroute $P_1$ to pass through $w_{e,v}$ instead of $w_{e,u}$ so as to obtain $P'_1$.

Consider the red vertices $r_{v,1}$ and $r_{v,2}$ for an AND vertex $v \in V(H)$ that is incident to one weight-$2$ edge $e$ and two weight-$1$ edges $f$ and $g$.
Then, $P_1$ must pass through either $k_v$ or $w_{f,v}, w_{g,v}$.
The same observation applies to $P'_1$, too.
If $P_1$ passes through $k_v$ and $P_2$ passes through neither $w_{f,v}$ nor $w_{g,v}$, then we may reroute $P_1$ to pass through $w_{f,v}, w_{g,v}$ so as to obtain $P'_1$.
If $P_1$ passes through $w_{f,v}$ and $w_{g,v}$, and $P_2$ does not pass through $k_v$, then we may reroute $P_1$ to pass through $k_v$ so as to obtain $P'_1$.

Rerouting within an AND vertex gadget does not correspond to flips in an orientation.
On the other hand, rerouting within an edge gadget corresponds to a flip in an orientation.
Namely, rerouting $P_1$ to pass through $w_{e,v}$ instead of $w_{e,u}$ is mapped to a flip of the direction of $e$.
If one rerouting operation involves the reconnection within several edge gadgets, then the operation is mapped to a set of flips of the corresponding edges.

We now argue that such a set of flips can be performed sequentially.
To this end, imagine that one rerouting operation of $P_1$ to $P'_1$ is decomposed into a number of reconfiguration steps each of which involves only one edge gadget.
Then, the sequence of those decomposed reconfiguration steps together yields $P'_1$ from $P_1$, and each step corresponds to a single flip in $H$.
Since two or more edge gadgets do not share a vertex in $G$, such a decomposition is well-defined, and we are done.

Similarly, we may study the case where we transform $P_2$ to $P'_2$ in such a way that $P_1$ and $P'_2$ are internally vertex-disjoint $s$-$t$ paths.
Rerouting within an AND vertex gadget and an OR vertex gadget is similar, and does not correspond to a flip in an orientation.
On the other hand, rerouting within an edge gadget corresponds to a flip of the corresponding edge.

When we are able to reroute a path within an edge gadget $G_e$ for an edge $e=\{u,v\} \in E(H)$, the other path passes through neither $w_{e,u}$ nor $w_{e,v}$.
This means that, even without $e$, the total weights of incoming arcs to $u$ and $v$ are at least two, respectively.
Therefore, the flip described above maintains the property that the resulting orientation is again an NCL configuration.

There can be several gadgets that are involved in one reconfiguration step. However, the number of such gadgets is at most $|V(H)|+|E(H)|$.
Hence, that reconfiguration step corresponds to at most $|V(H)|+|E(H)|$ flips in $H$.
\end{proof}

Let $\sigma'$ be the canonical NCL configuration in Claim \ref{clm:pspacecompl-twopaths-disjpaths-to-NCL}, and $\mathcal{P}'' = (P''_1, P''_2)$ be the family of two internally vertex-disjoint $s$-$t$ paths of $G$ that is created from $\sigma'$.
Note that $\mathcal{P}''$ does not have to be identical to $\mathcal{P}'$.
However, $\mathcal{P}''$ can be obtained from $\mathcal{P}'$ by at most $|V(H)|+|V^{\mathrm{a}}(H)|$ reconfiguration steps (recall that $V^{\mathrm{a}}(H)$ is the set of AND vertices of $H$).
This is because the situations in edge gadgets are identical in $\mathcal{P}'$ and $\mathcal{P}''$, and reconfiguration steps would be needed only around vertex gadgets.

As a summary, we have proved that the reduction is sound, complete, and polynomially bounded.
See \figurename~\ref{fig:pspacecompl2}.
We now complete the proof of Theorem \ref{thm:pspacecompl-twopaths}.
Let $\sigma$ and $\tau$ be two NCL configurations.
Then, define two $s$-$t$ linkages $\mathcal{P} = (P_1, P_2)$ and $\mathcal{Q} = (Q_1, Q_2)$ of $G$ as they are created from $\sigma$ and $\tau$, respectively.
Suppose that we may transform $\sigma$ to $\tau$ by a sequence of flips in such a way that all the intermediate orientations are NCL configurations.
Then, the correspondence above gives a sequence of reconfiguration steps 
to transform $\mathcal{P}$ to $\mathcal{Q}$ so that all the intermediate sets of two $s$-$t$ paths are internally vertex-disjoint.
On the other hand, suppose that we may transform $\mathcal{P}$ to $\mathcal{Q}$ by reconfiguration steps.
Then, the correspondence above gives a sequence of flips to transform $\sigma$ to $\tau$ so that all the intermediate orientations are NCL configurations.
\end{proof}

\begin{figure}[pt]
    \centering
    \includegraphics[width=.7\textwidth]{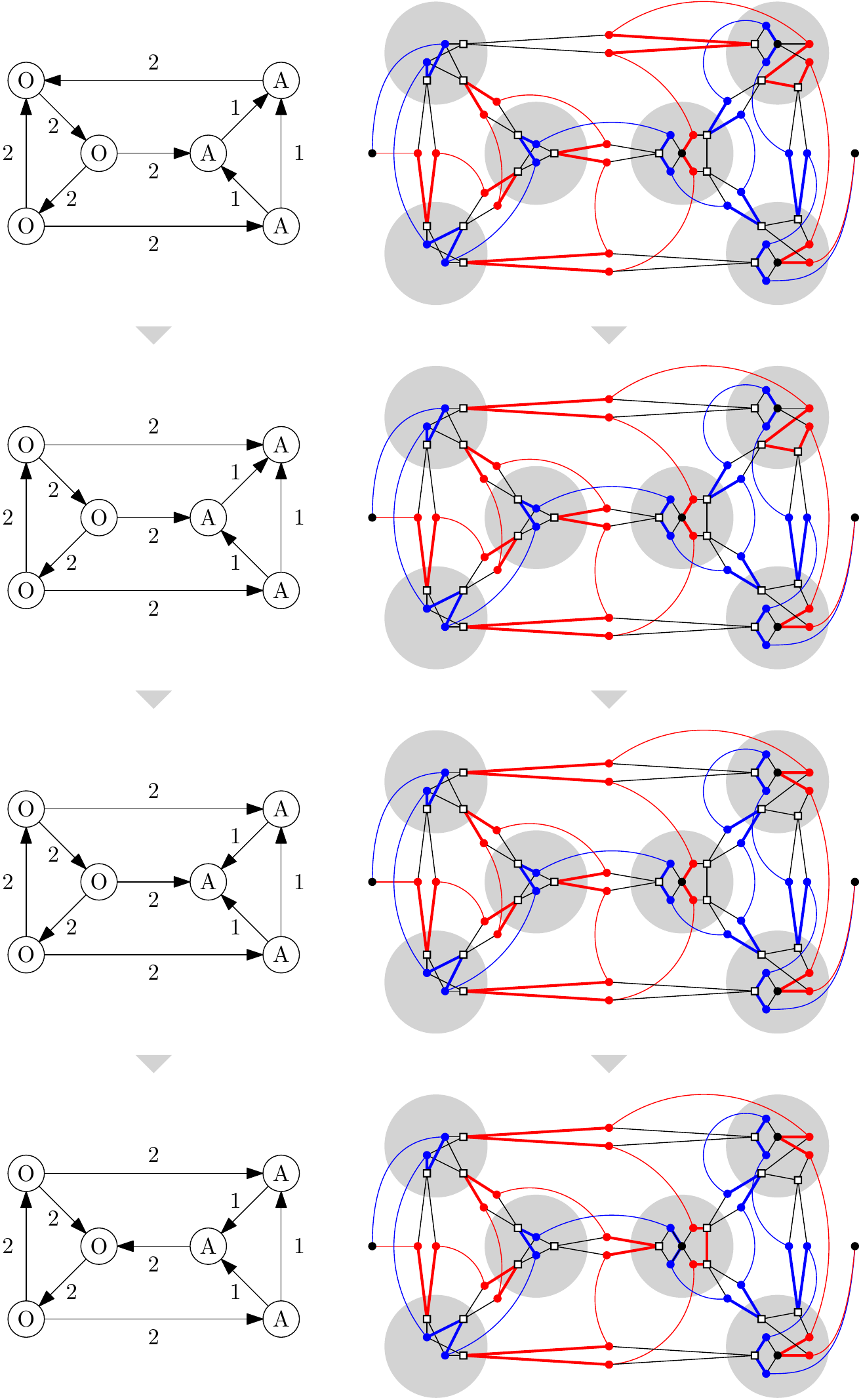}
    \caption{The correspondence between flips in NCL configurations and reconfiguration steps for vertex-disjoint paths.}
    \label{fig:pspacecompl2}
\end{figure}

It is useful to note that the reduction can be modified so that $G$ in the constructed instance of \stDPR has bounded bandwidth if the AND/OR graph $H$ in a given instance of the NCL reconfiguration has bounded bandwidth.
A brief sketch is given in Appendix \ref{sec:stDPR-bandwidth}.
This shows that Theorem \ref{thm:pspacecompl-twopaths} holds even for graphs of bounded bandwidth and maximum degree four.

\subsection{Planar Graphs of Bounded Bandwidth with Unbounded $k$}

A similar reduction proves the $\PSPACE$-completeness for planar graphs of bounded bandwidth when the number of paths is unbounded.

\planarhardness*

\begin{proof}
We have already observed that the problem is in $\PSPACE$. To show the $\PSPACE$-hardness, we now give a transformation of a given planar AND/OR graph $H=(V(H),E(H))$ of bounded bandwidth to an undirected planar graph $G=(V(G),E(G))$ of bounded bandwidth.
Each edge $e = \{u, v\} \in E(H)$ of $H$ is mapped to the following \emph{edge gadget} $G_e$ (see \figurename~\ref{fig:planar_pspacecompl_edge1}):
\begin{align*}
    V(G_e) &= \{w_{e,u}, w_{e,v}, s_{e}, t_{e}\},\\
    E(G_e) &= \{\{w_{e,u}, s_{e}\}, \{w_{e,u}, t_{e}\}, \{w_{e,v}, s_{e}\}, \{w_{e,v}, t_{e}\}\}.
\end{align*}

\begin{figure}[t]
    \centering
    \includegraphics{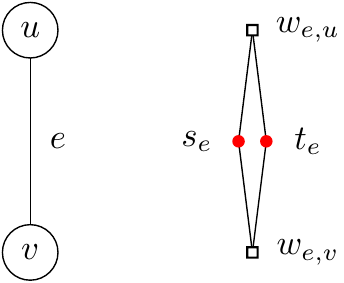}
    \caption{An edge gadget.}
    \label{fig:planar_pspacecompl_edge1}
\end{figure}

Let $v \in V(H)$ be an OR vertex, and $e,f,g \in E(H)$ be three weight-$2$ edges incident to $v$.
Then, $v$ is mapped to the following \emph{OR vertex gadget} $G_v$ (see \figurename~\ref{fig:planar_pspacecompl_overt1}):
\begin{align*}
    V(G_v) &= \{w_{e,v}, w_{f,v}, w_{g,v}, a_v, b_v, c_v, d_v, s_v, t_v, s^{\mathrm{o}}_v, t^{\mathrm{o}}_v\},\\
    E(G_v) &= \{\{w_{e,v}, a_v\}, \{w_{e,v}, b_v\}, \{w_{e,v}, t_v\}, \{w_{f,v}, a_v\}, \{w_{f,v}, s_v\},\\
        &\qquad \{w_{g,v}, b_v\}, \{w_{g,v}, s_v\}, \{a_v, c_v\}, \{a_v, d_v\}, \{b_v, c_v\}, \{b_v, d_v\},\\
        &\qquad \{c_v, s_v\}, \{c_v, s^{\mathrm{o}}_v\}, \{c_v, t^{\mathrm{o}}_v\}, \{d_v, t_v\}, \{d_v, s^{\mathrm{o}}_v\}, \{d_v, t^{\mathrm{o}}_v\}\}.
\end{align*}
Note that the vertices with the same label are identified.

\begin{figure}[t]
    \centering
    \includegraphics{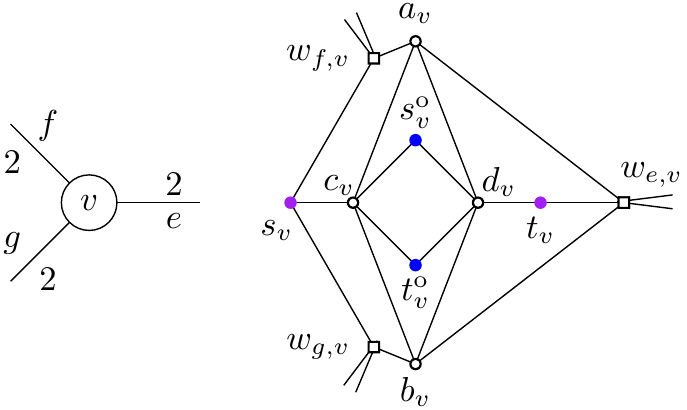}
    \caption{An OR vertex gadget.}
    \label{fig:planar_pspacecompl_overt1}
\end{figure}

An AND vertex $v \in V(H)$ that is incident to one weight-$2$ edge $e$ and two weight-$1$ edges $f, g$ is mapped to the following \emph{AND vertex gadget} $G_v$  (see \figurename~\ref{fig:planar_pspacecompl_avert1}):
\begin{align*}
    V(G_v) &= \{w_{e,v}, w_{f,v}, w_{g,v}, s_v, t_v\},\\
    E(G_v) &= \{ \{w_{e,v}, s_v\},
                 \{w_{e,v}, t_v\},
                 \{w_{f,v}, s_v\},
                 \{w_{g,v}, t_v\},
                 \{w_{f,v}, w_{g,v}\}\}.
\end{align*}

\begin{figure}[t]
    \centering
    \includegraphics{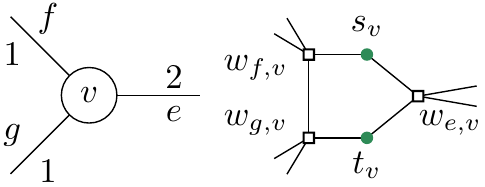}
    \caption{An AND vertex gadget.}
    \label{fig:planar_pspacecompl_avert1}
\end{figure}

Define
\begin{align*}
    V(G) &= \bigcup_{v \in V(H)} V(G_v) \cup \bigcup_{e \in E(H)} V(G_e),\\
    E(G) &= \bigcup_{v \in V(H)} E(G_v) \cup \bigcup_{e \in E(H)}E(G_e).
\end{align*}
This completes the construction of $G$.
\figurename~\ref{fig:planar_pspacecompl1} may help the reader understand.

We note that if $H$ is planar and of bounded bandwidth, then so is $G$.
To see that $G$ is planar, observe that each gadget can be drawn on the plane without edge crossing in such a way that the vertices shared with other gadgets are placed on its outer face.
Then, a plane drawing of $G$ can be obtained from a plane drawing of $H$ by replacing each vertex and edge of $H$ by the corresponding gadget.
To see that the bandwidth of $G$ is bounded, observe that each gadget is of constant size.
Therefore, by giving a slack of constant width $c$ for each vertex of $H$, we may obtain an injective mapping of the vertices of $G$ to $\Z$ whose bandwidth is at most $c$ times the bandwidth of $H$.\footnote{The constant $c$ can be chosen as $16$, but this particular choice is not relevant.}

\begin{figure}[t]
    \centering
    \includegraphics{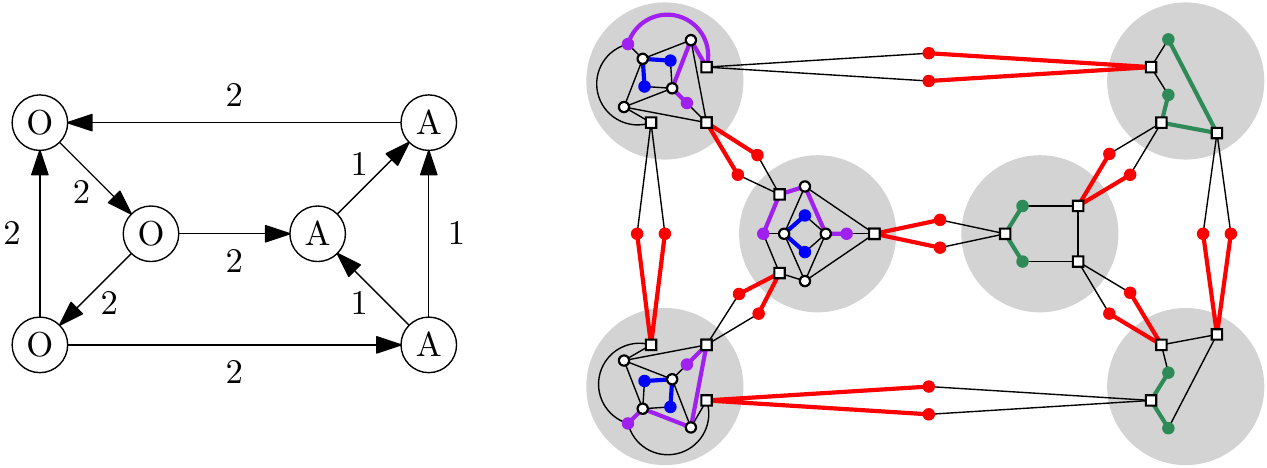}
    \caption{Mapping an NCL configuration to a linkage. In the NCL configuration, the vertex label A stands for an AND vertex, and the vertex label O stands for an OR vertex.}
    \label{fig:planar_pspacecompl1}
\end{figure}

In the constructed graph $G$, we consider the following pairs of vertices:
\[
    \{(s_e, t_e) \mid e \in E(H)\} \cup \{(s_v, t_v) \mid v \in V(H)\} \cup \{(s^{\mathrm{o}}_v, t^{\mathrm{o}}_v) \mid v \in V^{\mathrm{o}}(H) \},
\]
where $V^{\mathrm{o}}(H)$ denotes the set of OR vertices in $H$.
The number of pairs is $|E(H)|+|V(H)|+|V^{\mathrm{o}}(H)|$.

We now describe the way how an NCL configuration $\sigma$ of $H$ translates to a linkage
\[
    \mathcal{P} = \{P_e \mid e \in E(H)\} \cup \{P_v \mid v \in V(H)\} \cup \{P^{\mathrm{o}}_v \mid v \in V^{\mathrm{o}}(H)\}
\] 
of $G$, where $P_e$ joins $s_e$ and $t_e$ for every $e \in E(H)$, $P_v$ joins $s_v$ and $t_v$ for every $v \in V(H)$, and $P^{\mathrm{o}}_v$ joins $s^{\mathrm{o}}_v$ and $t^{\mathrm{o}}_v$ for every $v \in V^{\mathrm{o}}(H)$.
For convenience, $P_e$ is called a \emph{red} path for every $e \in E(H)$, and $P^{\mathrm{o}}_v$ is called a \emph{blue} path for every $v \in V^{\mathrm{o}}(H)$.
Similarly, for every $v \in V(H)$, $P_v$ is called a \emph{green} path if $v$ is an AND vertex, and a \emph{purple} path if $v$ is an OR vertex.

First, we specify the red paths $P_e$ for all $e=\{u,v\} \in E(H)$.
The path $P_e$ connects $s_e$ and $t_e$ and its length is two.
Therefore, there will be a unique middle vertex.
If the edge $e$ is oriented from $u$ to $v$ in the NCL configuration $\sigma$, then the path $P_e$ visits $w_{e,u}$ as its middle vertex.
If $e$ is oriented from $v$ to $u$ in $\sigma$, then $P_e$ visits $w_{e,v}$ as its middle vertex.
This finishes the specification of the red paths $P_e$.

Next, we specify the green paths $P_v$ that connect $s_v$ and $t_v$ for all AND vertices $v \in V(H)$. 
Recall that $v$ is incident to a weight-$2$ edge $e$ and two weight-$1$ edges $f, g$.
If $e$ is oriented toward $v$, then $P_v$ passes through $w_{e,v}$ and its length is two.
Otherwise, $f$ and $g$ must be oriented toward $v$.
Then, $P_v$ passes through $w_{f,v}$ and $w_{g,v}$ and its length is three.
This finishes the specification of the green paths $P_v$.

Finally, we specify the purple paths $P_v$ that connect $s_v$ and $t_v$ and the blue paths $P^{\mathrm{o}}_v$ that connect $s^{\mathrm{o}}_v$ and $t^{\mathrm{o}}_v$ for all OR vertices $v \in V^{\mathrm{o}}(H)$.
Recall that $v$ is incident to three weight-$2$ edges $e, f, g$.
We have three cases, which are not exclusive.
If $e$ is oriented toward $v$, then the purple path $P_v$ passes through $c_v, a_v, w_{e,v}$ (or $c_v, b_v, w_{e,v}$), and the blue path $P^{\mathrm{o}}_v$ passes through $d_v$.
If $f$ is oriented toward $v$, then $P_v$ passes through $w_{f,v}$, $a_v$, $d_v$, and $P^{\mathrm{o}}_v$ passes through $c_v$.
If $g$ is oriented toward $v$, then $P_v$ passes through $w_{g,v}$, $b_v$, $d_v$, and $P^{\mathrm{o}}_v$ passes through $c_v$.
This finishes the specification of the purple paths $P_v$ and the blue paths $P^{\mathrm{o}}_v$.
Note that the lengths of purple paths are always four and the lengths of blue paths are always two.

This finishes the construction of $\mathcal{P}$.
With this construction, we say that the linkage $\mathcal{P}$ \emph{is created from} the NCL configuration $\sigma$.

\begin{claim}
\label{clm:pspacecompl-planar-vdisj}
The constructed paths are vertex-disjoint.
\end{claim}
\begin{proof}
Each path visits vertices in the corresponding gadget only, and within each OR gadget the two constructed paths do not share vertices by the construction rule.
Therefore, it suffices to examine the situation at identified vertices between an edge gadget and an AND/OR vertex gadget.

Let $v$ be an OR vertex that are incident to three weight-$2$ edges $e, f, g$.
Assume that the corresponding OR vertex gadget is constructed as above.

Suppose that the vertex $w_{e,v}$ is shared by two of the constructed paths.
One of those paths is the red path $P_e$ that connects $s_e$ and $t_e$; the other path is the purple path $P_v$ that connects $s_v$ and $t_v$.
Since $P_e$ passes through $w_{e,v}$, the edge $e$ is not oriented toward $v$.
On the other hand, since $P_v$ passes through $w_{e,v}$, the edge $e$ is oriented toward $v$.
This is a contradiction, and thus $w_{e,v}$ is not shared by two paths.

Suppose that the vertex $w_{f,v}$ is shared by two of the constructed paths.
One of those paths is the red path $P_f$, which means that the edge $f$ is not oriented toward $v$.
Then, the other path is the purple path $P_v$.
By the construction, it holds that $f$ is oriented toward $v$.
This is a contradiction.
The same conclusion holds when $w_{g,v}$ is shared by two paths.

Let $v$ be an AND vertex that are incident to one weight-$2$ edge $e$ and two weight-$1$ edges $f, g$.
Assume that the corresponding AND vertex gadget is constructed as above.

Suppose that the vertex $w_{e,v}$ is shared by two of the constructed paths.
One of those paths is the red path $P_e$ that connects $s_e$ and $t_e$; the other path is the green path $P_v$ that connects $s_v$ and $t_v$.
Since $P_e$ passes through $w_{e,v}$, the edge $e$ is not oriented toward $v$.
On the other hand, since $P_v$ passes through $w_{e,v}$, the edge $e$ is oriented toward $v$.
This is a contradiction, and thus $w_{e,v}$ is not shared by two paths.

Suppose that the vertex $w_{f,v}$ is shared by two of the constructed paths.
One of those paths is the red path $P_f$ that connects $s_f$ and $t_f$, which means that the edge $f$ is not oriented toward $v$.
The other path is the green path $P_v$ that connects $s_v$ and $t_v$, which means that $f$ is oriented toward $v$.
This is a contradiction.
The same conclusion holds when $w_{g,v}$ is shared by two paths.

In summary, there is no vertex in the instance $G$ that is shared by two or more constructed paths.
\end{proof}

We now study how a flip in an orientation maps to a sequence of reconfiguration steps of paths.
\begin{claim}
\label{clm:pspacecompl-planar-NCL-to-disjpaths}
Let $\mathcal{P}$ be a linkage of $G$ that is created from an NCL configuration $\sigma$.
If an NCL configuration $\sigma'$ is obtained from $\sigma$ by an flip of a single edge, then a linkage $\mathcal{P}'$ of $G$ that is created from $\sigma'$ can be obtained from $\mathcal{P}$ by at most a constant number of reconfiguration steps in $G$.
\end{claim}
\begin{proof}
Assume that we are given an NCL configuration $\sigma$ of $H$.
Let
\[
\mathcal{P} = \{ P_e \mid e \in E(H)\} \cup \{P_v \mid v \in V(H)\} \cup \{ P^{\mathrm{o}}_v \mid v \in V^{\mathrm{o}}(H)\}
\]
be a linkage of $G$ that is created from $\sigma$.
By the construction, each of the colored paths passes through all the vertices of the corresponding color.

As the first case, consider a flip of a weight-$1$ edge $f$ that connects two AND vertices $u,v$ from $(u,v)$ to $(v,u)$.
Then, the red path $P_f$ passes through $s_f, w_{f,u}, t_f$ of the weight-$1$ edge gadget in this order.
The path $P'_f$ is obtained by replacing it with $s_f, w_{f,v}, t_f$.

This is possible only if $w_{f,v}$ is not passes through by any other paths; we now observe this is indeed the case.
Suppose that $w_{f,v}$ is passed through by another path.
Then, such a path should be a green path $P_v$ in the AND vertex gadget $G_v$.
When $P_v$ passes through $w_{f,v}$, a unique weight-$2$ edge $e$ incident to $v$ is not oriented toward $v$ in $\sigma$ by the construction.
Therefore, after flipping to $(v,u)$ the total sum of the incoming edges to $v$ will be $1$. This is a contradiction to the property of an NCL configuration.

As the second case, consider a flip of a weight-$2$ edge $e$ that connects a vertex $u$ and an AND vertex $v$ from $(u,v)$ to $(v,u)$.
Then, the red path $P_e$ passes through $s_e, w_{e,u}, t_e$ of the weight-$2$ edge gadget $G_e$ in this order.
The path $P'_e$ is obtained by replacing it with $s_e, w_{e,v}, t_e$.

This is possible only if $w_{e,v}$ is not passes through by an other paths.
However, since $e$ is oriented toward $v$, a green path $P_v$ in the AND vertex gadget $G_v$ passes through $w_{e,v}$. 
Thus, in order to make $w_{e,v}$ vacant, we apply the following operations before the above transformation. 
Since $e$ can be flipped, after the flip the total sum of the incoming edges to $v$ should be at least two, which means that the other two weight-$1$ edges $f$, $g$ that are incident to $v$ should be directed toward $v$.
Therefore, $w_{v,f}$ and $w_{v,g}$ are not occupied by any other paths in $\mathcal{P}$.
This implies that by rerouting the green path $P_v$ to pass through $w_{v,f}, w_{v,g}$ before rerouting $P_e$ to $P'_e$ we obtain the linkage that corresponds to $\sigma'$.

As the third but final case, consider a flip of a weight-$2$ edge $e$ that connects a vertex $u$ and an OR vertex $v$ from $(u,v)$ to $(v,u)$.
Then, the red path $P_e$ passes through $s_e, w_{e,u}, t_e$ of the edge gadget $G_e$ in this order.
The $P'_e$ is obtained by replacing it with $s_e, w_{e,v}, t_e$.

This is possible only if $w_{e,v}$ is not passed through by any other paths.
If this is the case, there is no problem.
However, in some cases, $w_{e,v}$ is indeed passed through by another path.
Then, such a path should be a purple path $P_v$ in the OR vertex gadget $G_v$.
By construction, the edge $e$ is oriented toward $v$ in $\sigma$.
Since $e$ can be flipped, after the flip the total sum of the incoming edges to $v$ should be at least two, which means that one of the other two weight-$2$ edges $f$, $g$ that are incident to $v$ should be directed toward $v$.
Therefore, $w_{v,f}$ or $w_{v,g}$ is not occupied by any other paths in $\mathcal{P}$.
This implies that by rerouting the purple path $P_v$ to pass through $w_{v,f}$ or $w_{v,g}$ before rerouting $P_e$ to $P'_e$, possibly with two more reconfiguration steps involving the blue path $P^{\mathrm{o}}_v$, we obtain the linkage that corresponds to $\sigma'$.

In the second and third cases above, if $u$ is an AND vertex, then 
we reroute a green path $P_u$ in the AND vertex gadget $G_u$ to pass through $w_{e,u}$. 
Then, we obtain a linkage $\mathcal{P}'$ of $G$ that is created from $\sigma'$. 
\end{proof}

To argue the reverse direction of correspondence, we describe a canonical way of constructing an NCL configuration $\sigma$ of $H$ from a linkage $\mathcal{P}$ of $G$.

Let $e=\{u,v\}$ be an edge of $H$.
We note that a path connecting $s_e$ and $t_e$ cannot go through both of $w_{e,u}$ and $w_{e,v}$ due to the vertex-disjointness of the paths in $\mathcal{P}$.
If $P_e$ goes through $w_{e,u}$, then the edge $e$ is oriented toward $v$ in $\sigma$.
If $P_e$ goes through $w_{e,v}$, then the edge $e$ is oriented toward $u$ in $\sigma$.
We call the obtained orientation $\sigma$ the \emph{canonical NCL configuration obtained from $\mathcal{P}$}.

\begin{claim}
Let $\mathcal{P}$ be a linkage of $G$.
Then, the canonical NCL configuration $\sigma$ obtained from $\mathcal{P}$ is indeed an NCL configuration of $H$.
\end{claim}
\begin{proof}
We check that $\sigma$ satisfies the weight requirement at every vertex $v$ of $H$.

Let $v$ be an OR vertex of $H$ that is incident to three weight-$2$ edges $e$, $f$ and $g$.
The corresponding OR vertex gadget is given as in \figurename~\ref{fig:planar_pspacecompl_overt1}.
Since the blue path in the gadget $G_v$ contains a vertex $c_v$ or $d_v$, 
the purple path in $G_v$ contains at least one of $w_{e,v}$, $w_{f,v}$ and $w_{g,v}$.
This shows that one of $w_{e,v}$, $w_{f,v}$ and $w_{g,v}$ is not occupied by any red paths.
Therefore, in the canonical NCL configuration $\sigma$, one of the edges $e$, $f$ and $g$ is oriented toward $v$.
This means that the weight requirement is satisfied at $v$.

Next, let $v$ be an AND vertex of $H$ that is incident to one weight-$2$ edge $e$ and two weight-$1$ edges $f$ and $g$.
The green path between $s_v$ to $t_v$ either goes through $w_{e,v}$ or $w_{f,v}, w_{g,v}$.
If it goes through $w_{e,v}$, then the red path $P_e$ does not go through $w_{e,v}$ due to vertex-disjointness.
In this case, the weight-$2$ edge $e$ is oriented toward $v$ in $\sigma$.
If the green path goes through $w_{f,v}$ and $w_{g,v}$, then the red paths $P_f$ and $P_g$ do not go through $w_{f,v}$ and $w_{g,v}$, respectively.
In this case, the two weight-$1$ edges $f$ and $g$ are oriented toward $v$ in $\sigma$.
Hence, in both cases, the weight requirement is satisfied at $v$.
\end{proof}

\begin{claim}
\label{clm:pspacecompl-planar-disjpaths-to-NCL}
Let $\mathcal{P}$ be a linkage of $G$ and $\sigma$ be the canonical NCL configuration obtained from $\mathcal{P}$.
Let a linkage $\mathcal{P}'$ be obtained by a single reconfiguration step from $\mathcal{P}$ and $\sigma'$ be the canonical NCL configuration obtained from $\mathcal{P}'$.
Then, $\sigma'$ can be obtained from $\sigma$ by at most one flip in $H$.
\end{claim}
\begin{proof}
We distinguish cases by colors of paths that are reconfigured by one step.

First, consider a case where $\mathcal{P}$ and $\mathcal{P}'$ differ by one red path associated with an edge gadget $G_e$.
Assume that $P_e$ is reconfigured to $P'_e$.
Without loss of generality, assume that $P_e$ passes through $s_e, w_{e,u}, t_e$, and $P'_e$ passes through $s_e, w_{e,v}, t_e$ in this order.
Then, to obtain $\sigma'$ from $\sigma$ we just need to flip the direction of the edge $e$. 

Next, consider a case where $\mathcal{P}$ and $\mathcal{P}'$ differ by one path that is not a red path.
Then, the canonical NCL configuration $\sigma'$ obtained from $\mathcal{P}'$ is identical to $\sigma$.
Therefore, we need no flip to construct $\sigma'$.
\end{proof}

Let $\sigma'$ be the canonical NCL configuration in Claim \ref{clm:pspacecompl-planar-disjpaths-to-NCL}, and $\mathcal{P}''$ be the linkage created from $\sigma'$.
Note that $\mathcal{P}''$ does not have to be identical to $\mathcal{P}'$.
However, $\mathcal{P}''$ can be obtained from $\mathcal{P}'$ by 
$O(|V(H)|)$ reconfiguration steps. 
This is because the situations in edge gadgets are identical in $\mathcal{P}'$ and $\mathcal{P}''$, and reconfiguration steps would be needed only around vertex gadgets.

As a summary, we have proved that the reduction is sound, complete, and polynomially bounded.
See \figurename~\ref{fig:planar_pspacecompl2}.
We now complete the proof of Theorem \ref{thm:pspacecompl-planar}.
Let $\sigma$ and $\tau$ be two NCL configurations.
Then, define two linkages $\mathcal{P}$ and $\mathcal{Q}$ as they correspond to $\sigma$ and $\tau$, respectively.
Suppose that we may transform $\sigma$ to $\tau$ by a sequence of flips in such a way that all the intermediate orientations are NCL configurations.
Then, the correspondence above gives a sequence of reconfiguration steps 
to transform $\mathcal{P}$ to $\mathcal{Q}$ so that all the intermediate families of paths are vertex-disjoint.
On the other hand, suppose that we may transform $\mathcal{P}$ to $\mathcal{Q}$ by reconfiguration steps.
Then, the correspondence above gives a sequence of flips to transform $\sigma$ to $\tau$ so that all the intermediate orientations are NCL configurations.
\end{proof}

\begin{figure}
    \centering
    \includegraphics[width=.7\textwidth]{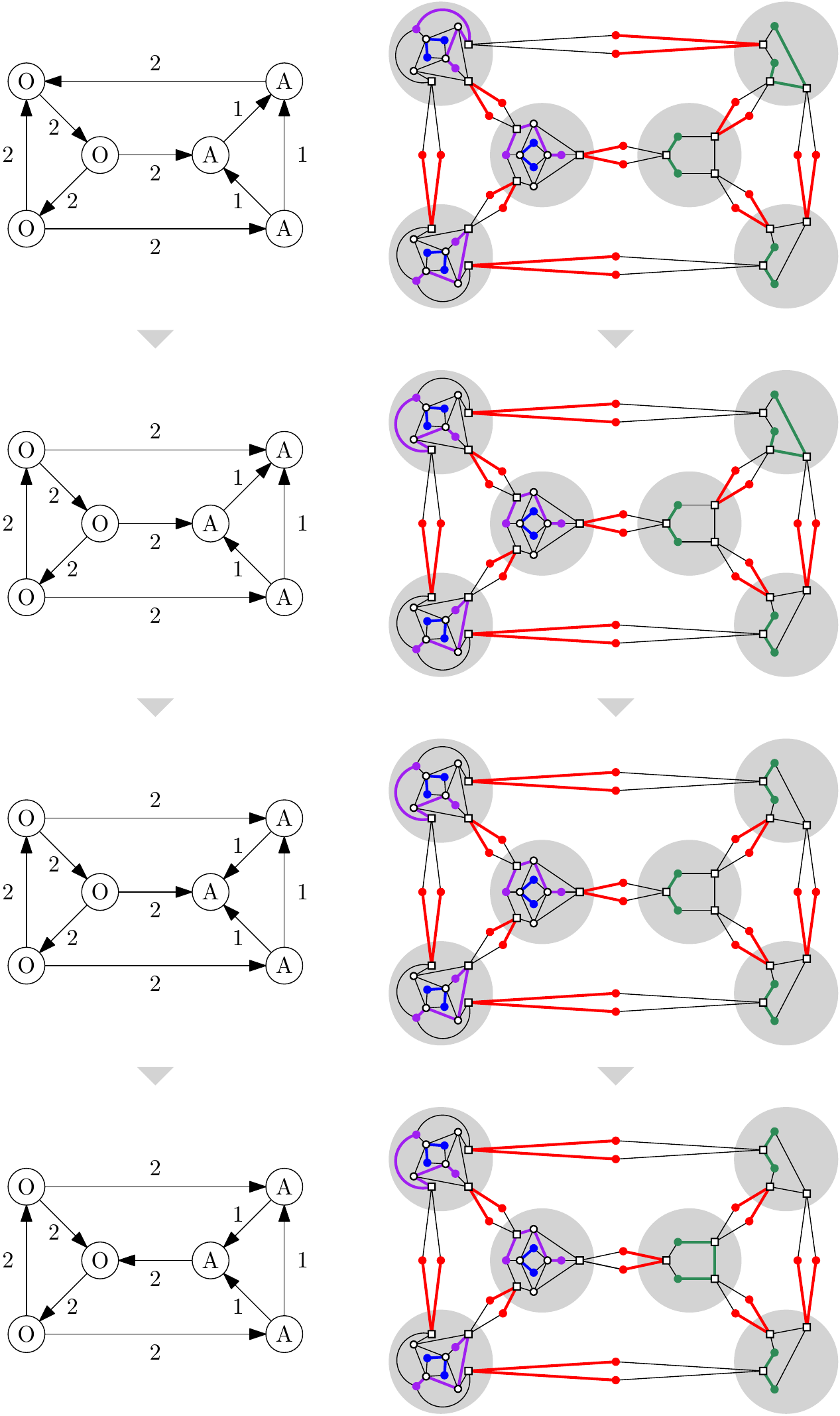}
    \caption{The correspondence between flips in NCL configurations and reconfiguration steps for vertex-disjoint paths.}
    \label{fig:planar_pspacecompl2}
\end{figure}

\section{Concluding Remarks}

We leave several open problems for future research.
We proved that \DPR 
can be solved in polynomial time when the problem is restricted to the two-face instances.
On the other hand, we do not know whether \DPR in planar graphs can be solved in polynomial time for fixed $k$, and even when $k=2$, if we drop the requirement that inputs are two-face instances.

We did not try to minimize the number of reconfiguration steps when a reconfiguration sequence exists.
It is an open problem whether a shortest reconfiguration sequence can be found in polynomial time for \DPR restricted to planar two-face instances.

A natural extension of our studies is to consider a higher-genus surface.
As a preliminary result, in Appendix \ref{sec:surfaceappendix}, we give a proof (sketch) to show that when the number $k$ of curves is two, the reconfiguration is \emph{always} possible for any connected orientable closed surface $\Sigma_g$ of genus $g \geq 1$.
Note that this result does not refer to graphs embedded on $\Sigma_g$, but only refers to the case when curves can pass through any points on the surface.
It is not clear what we can say for \DPR for graphs embedded on $\Sigma_g$, $g \geq 1$.


\appendix

\section{Curves on Closed Surfaces}
\label{sec:surfaceappendix}

Let $\Sigma_g$ denote a connected orientable closed surface of genus $g$.
For instance, $\Sigma_0$ is a 2-sphere $S^2$ and $\Sigma_1$ is a torus.
In contrast to the 2-sphere case, when $k=2$, the reconfiguration is always possible on $\Sigma_g$ for $g\geq 1$. 

\begin{theorem}
\label{thm:genus>0}
Let $\mathcal{P}=(P_1,P_2)$ and $\mathcal{Q}=(Q_1,Q_2)$ be linkages on $\Sigma_g$ for $g\geq 1$.
Then, $\mathcal{P}$ is always reconfigurable to $\mathcal{Q}$.
\end{theorem}

\begin{proof}[Proof Sketch]
A small perturbation allows us to assume that the curves intersect transversally.
The proof is by induction on the number of intersection points of $\mathcal{P}$ and $\mathcal{Q}$, say $n$.
The case $n=4$ is obvious since they intersect only at the endpoints.
Suppose that the statement holds up to $n-1$.
We first cut the surface $\Sigma_g$ along $\mathcal{P}$ and $\mathcal{Q}$.
Let $S$ be a connected component.
Then, $\partial S$ is a disjoint union of alternating cycles of edges derived from $\mathcal{P}$ and $\mathcal{Q}$.
If $\partial S$ has a $2$-cycle, then one can eliminate this cycle by a single step.
On the other hand, if $\partial S$ has at least six edges, then there are two edges derived from $P_i$ for some $i=1,2$, and thus one can find another $\mathcal{P}'$ such that the number of intersections between the curves in $\mathcal{P}'$ and $\mathcal{Q}$ is less than $n$.

Therefore, we shall consider the case where $\partial S$ is a $4$-cycle and show that there is a connected component $S$ such that the genus of the closed surface obtained by capping its boundary with disks is at least $1$.
Then, one finds a desired $\mathcal{P}'$ after some steps.
Assume that there is no such an $S$, that is, all connected components are disks, and hence we have a cell decomposition of $\Sigma_g$.
Let $m$ be the number of $2$-cells, namely disks.
Then, the numbers of $1$- and $0$-cells are $2m$ and $m+2$, respectively.
It follows that the Euler characteristic of $\Sigma_g$ is $(m+2)-2m+m=2$, and thus $g=0$ (see \cite[Section~1.1.1]{FaMa12} for example).
This is a contradiction.
\end{proof}

\section{$\mu(P_i, Q_j)$ Takes the Same Value}
\label{sec:2facesamevalue}

In this section, we only consider piecewise smooth curves and paths on a plane. 
For an oriented curve $C$, let $\overline{C}$ denote the curve with the opposite orientation. 
We first summarize basic properties of $\mu$ that follow from definition. 
\begin{enumerate}
\item[(A1)] $\mu(C_1,C_2)=-\mu(C_2,C_1)=-\mu(C_1,\overline{C}_2)$ for any oriented curves $C_1$ and $C_2$. 
\item[(A2)] Suppose that $C_1 \cup C_2$ is an oriented curve obtained by concatenating two oriented curves $C_1$ and $C_2$, and let $C_3$ be another oriented curve. 
Then, $\mu(C_1 \cup C_2, C_3)= \mu(C_1,C_3) + \mu(C_2,C_3)$ and $\mu(C_3, C_1 \cup C_2)= \mu(C_3,C_1) + \mu(C_3,C_2)$. 
\item[(A3)] $\mu(C_1,C_2)=0$ for any oriented closed curves $C_1$ and $C_2$ (by the Jordan curve theorem). 
\end{enumerate}

By using these properties, we show that $\mu(P_i, Q_j)$ takes the same value in a two-face instance.

\begin{lemma}
\label{lem:2facesamevalue}
In a two-face instance $(G, \mathcal{P}, \mathcal{Q})$ of \DPR, 
$\mu(P_i, Q_j)$ takes the same value for any distinct $i, j \in [k]$. 
\end{lemma}

\begin{proof}
Let $s_{k+1}$ and $s_{k+2}$ be points on $\bd{S}$ such that $s_k, s_{k+1}, s_{k+2}$, and $s_1$ lie on $\partial S$ clockwise in this order. 
Similarly, let $t_{k+1}$ and $t_{k+2}$ be points on $\bd{T}$ such that $t_k, t_{k+1}, t_{k+2}$, and $t_1$ lie on $\partial T$ clockwise in this order. 
Then, there exist simple curves $P_{k+1}$ from $s_{k+1}$ to $t_{k+1}$ and $P_{k+2}$ from $s_{k+2}$ to $t_{k+2}$ 
such that $\{P_1, \dots , P_k, P_{k+1}, P_{k+2}\}$ forms a linkage in $\R^2 \setminus (S \cup T)$. 
Similarly, let $Q_{k+1}$ and $Q_{k+2}$ be simple curves such that $\{Q_1, \dots , Q_k, Q_{k+1}, Q_{k+2}\}$ forms a linkage in $\R^2 \setminus (S \cup T)$. 
In what follows, we show that $\mu(P_i, Q_j) = \mu(P_{k+1}, Q_{k+2})$ for any distinct $i,j \in [k]$. 

Let $i, j \in [k]$ with $i \neq j$. 
Let $C_s$ be a curve in $S$ from $s_{k+1}$ to $s_i$, and 
let $C_t$ be a curve in $T$ from $t_i$ to $t_{k+1}$. 
Since $P_i \cup C_t \cup \overline{P_{k+1}} \cup C_s$ and $\overline{P_j} \cup Q_j$ are oriented closed curves, by (A3), 
we obtain $\mu(P_i \cup C_t \cup \overline{P_{k+1}} \cup C_s, \overline{P_j} \cup Q_j) = 0$.  
Since $\overline{P_j} \cap (P_i \cup C_t \cup \overline{P_{k+1}} \cup C_s) = \emptyset$ and $Q_j \cap (C_s \cup C_t) = \emptyset$, 
this together with (A2) shows that $\mu(P_i, Q_j)  + \mu(\overline{P_{k+1}}, Q_j) = 0$. 
By (A1), we obtain $\mu(P_i, Q_j)  = \mu(P_{k+1}, Q_j)$. 

By the same argument, we can show that $\mu(P_{k+1}, Q_j) = \mu(P_{k+1}, Q_{k+2})$. 
Therefore, we obtain $\mu(P_i, Q_j)  =\mu(P_{k+1}, Q_{k+2})$ for any distinct $i, j \in [k]$, 
which completes the proof. 
\end{proof}

\section{Variation of the Reduction for Theorem \ref{thm:pspacecompl-twopaths}.}
\label{sec:stDPR-bandwidth}

In this section, we briefly sketch the modification of the proof for Theorem \ref{thm:pspacecompl-twopaths} to show that the $\PSPACE$-completeness of \stDPR also holds for graphs of bounded bandwidth.

The basic line of reduction is identical to the proof for Theorem \ref{thm:pspacecompl-twopaths}.
We require the constructed graph $G$ to have bounded bandwidth if the AND/OR graph $H$ in a given instance of the NCL reconfiguration has bounded bandwidth.

To ensure that $G$ has bounded bandwidth, we only change the construction of two paths, the red path $P_1$ and the blue path $P_2$.
In the original proof of Theorem \ref{thm:pspacecompl-twopaths}, we have a choice of the ordering along which $P_1$ goes through all the AND vertex gadgets and weight-$2$ edge gadgets, and $P_2$ goes through all the AND vertex gadgets, OR vertex gadgets and weight-$1$ edge gadgets.
We will exploit this freedom.

Since $H$ has bounded bandwidth, there exists an injective map $\pi\colon V(H) \to \mathbb{Z}$ such that $|\pi(u) - \pi(v)| \leq c$ for some constant $c$.
Let $H^{\circ}$ be the graph that is obtained from $H$ by subdividing each edge of $H$, i.e., replacing each edge of $H$ by a path of length two.
We construct an injective map $\pi^\circ\colon V(H^\circ) \to \mathbb{Z}$ as follows.
For each vertex $v \in V(H)$, set $\pi^\circ(v) = 4\pi(v)$.
For a vertex $w \in V(H^\circ)$ that is used for subdividing an edge $\{u,v\} \in E(H)$, where $\pi(u) < \pi(v)$, we set $\pi^\circ(w)$ as $\pi^\circ(w) \in \{\pi(u)+1, \pi(u)+2, \pi(u)+3\}$.
This is possible since the AND/OR graph $H$ is $3$-regular.
Observe that the bandwidth of $\pi^\circ$ is at most $4c+3$.
For more information on the bandwidth of graphs obtained by graph operations, see \cite{DBLP:journals/dm/ChvatalovaO86}.

Now, we replace each vertex of $H^{\circ}$ by the corresponding gadget given in the proof of Theorem \ref{thm:pspacecompl-twopaths}.
To ease the presentation, we add one isolated blue vertex to each weight-$2$ edge gadget, one isolated red vertex to each weight-$1$ edge gadget, and one isolated red vertex to each OR vertex gadget so that each gadget contains at least one red and at least one blue vertices.
Note that each gadget is of constant order (at most eight), even after adding isolated vertices.
Therefore, this replacement may increase the bandwidth at most by the factor of eight.
Denote by $V$ be the constructed vertex set and by $\pi'$ the implied injective map of bandwidth at most $8\cdot (4c+3) + 7$, 
where $\pi'$ satisfies that $8 \pi^\circ (w) \le \pi'(v) \le 8 \pi^\circ (w) + 7$ if $v \in V$ is contained in the gadget corresponding to $w \in V(H^\circ)$.

We insert $s$ and $t$ to $V$ and set $\pi'(s) = \min\{\pi'(v)\}-1$, $\pi'(t) = \max\{\pi'(v)\}+1$.
To wire the red path $P_1$, we follow the increasing order of the values of $\pi'$.
Namely, the path $P_1$ starts at $s$, visits the red vertices in the increasing order of $\pi'$ (with the routing within each gadget as the proof of Theorem \ref{thm:pspacecompl-twopaths}), and terminates at $t$.
Since each gadget contains a red vertex, each edge $\{x,y\}$ of $P_1$ satisfies 
$|\pi'(x)-\pi'(y)| \leq 15$. 
Wiring the blue path $P_2$ is similarly done.
This completes the whole reduction.

The constructed instance has bounded bandwidth, and the correctness is immediate from the proof of Theorem \ref{thm:pspacecompl-twopaths} since we only makes use of the freedom of ordering for the construction of $P_1$ and $P_2$.

Thus, \stDPR is $\PSPACE$-complete even when $k=2$ and $G$ has bounded bandwidth and maximum degree four.

\end{document}